\documentclass[letterpaper,twocolumn,10pt]{article}
\usepackage{usenix2019_v3}

\usepackage{amsmath,amssymb,amsfonts,amsthm}
\usepackage{graphicx}
\usepackage{booktabs}
\usepackage{tabularx}
\usepackage{array}
\usepackage{float}
\usepackage{listings}
\usepackage{algorithm}
\usepackage{algpseudocode}
\usepackage{subcaption}
\usepackage{xspace}
\usepackage{xcolor}
\usepackage{tikz}
\usetikzlibrary{arrows.meta, backgrounds, positioning, calc, shapes.geometric, fit, shadows, matrix, patterns, decorations.pathreplacing}
\microtypesetup{spacing=false}
\AtBeginDocument{\DeclareMathAlphabet{\mathcal}{OMS}{cmsy}{m}{n}}

\newtheorem{theorem}{Theorem}
\newtheorem{definition}{Definition}
\newtheorem{proposition}{Proposition}
\newtheorem{corollary}{Corollary}
\newtheorem{lemma}{Lemma}

\definecolor{slate}{RGB}{112,128,144}
\definecolor{emerald}{RGB}{16,163,127}
\definecolor{navy}{RGB}{24,49,83}
\definecolor{crimson}{RGB}{190,30,45}
\definecolor{amber}{RGB}{217,119,6}
\definecolor{cobalt}{RGB}{29,78,216}
\definecolor{teal}{RGB}{13,148,136}
\definecolor{purple}{RGB}{126,34,206}
\definecolor{softgray}{RGB}{247,249,252}
\definecolor{bordergray}{RGB}{218,224,233}
\definecolor{darkslate}{RGB}{60,72,88}
\definecolor{linegray}{RGB}{140,155,175}

\definecolor{crimsoncol}{RGB}{205,32,44}
\definecolor{emeraldcol}{RGB}{16,128,67}
\definecolor{cobaltcol}{RGB}{18,92,194}
\definecolor{navycol}{RGB}{24,43,73}
\definecolor{ambercol}{RGB}{217,130,43}
\definecolor{orangecol}{RGB}{225,105,35}
\definecolor{slatecol}{RGB}{88,108,138}

\hypersetup{
  pdftitle={The Illusion of Independent Quorums: Epistemic Fault Domains and Correlated Cognitive Failures in Agentic Quorums},
  pdfauthor={Jun He; Deying Yu},
  pdfsubject={Correlated cognitive failures and epistemic fault domains in agentic quorums},
  pdfkeywords={epistemic fault domains, correlated failures, agentic quorums, Byzantine fault tolerance, post-deterministic distributed systems}
}

\begin{document}

\title{\bf The Illusion of Independent Quorums:\\Epistemic Fault Domains and Correlated Cognitive Failures in Agentic Quorums}

\author{
  {\rm Jun He}\\
  OpenKedge.io
  \and
  {\rm Deying Yu}\\
  OpenKedge.io
}

\maketitle

\begin{abstract}
Multi-agent quorums are widely used to authorize high-stakes infrastructure and policy mutations, yet distinct reviewers often share upstream telemetry, documents, or tool backends. When upstream inputs fail, multiple votes collapse onto a single corrupted cause: \textbf{replication does not imply epistemic redundancy}. We introduce \emph{Epistemic Fault Domains} (EFDs) and the \emph{Structural Epistemic Cut} $\kappa_E$, which quantifies the minimum number of modeled root faults whose exposure covers an authorizing coalition relative to an explicit Epistemic Fault Basis. Under closed causal accounting, conservative exposure, and authorization alignment, $\kappa_E$ lower-bounds the number of roots required for semantic compromise ($\kappa_S$). We prove that arbitrarily large quorums can retain $\kappa_E=1$, that recognizing shared ancestry never increases credited resilience, and that adding voters at a fixed threshold cannot increase the cut under compatible exposure extensions. Finally, we design the Dependency-Aware Quorum Controller (DAQC) to enforce structural cuts at runtime admission, evaluate its mechanics via analytical derivations and simulations, and provide a frozen 120-task external benchmark suite.
\end{abstract}

\section{Introduction}
\label{sec:intro}

Autonomous agents increasingly execute high-impact mutations, including decommissioning database replicas, altering network security policies, granting emergency privileges, and initiating deployment rollbacks. Because these actions mutate durable system state, architectures commonly mandate threshold approvals, debate, or ensemble voting before execution \cite{wu2023autogen,hong2023metagpt,wang2024mixture,du2023improving,liang2023encouraging}. Multiple approvals are intended to provide redundancy, but multiple votes do not automatically provide multiple independent sources of evidence.

Reviewers frequently query the same stale telemetry cache, retrieve the same corrupted document, invoke the same faulty tool, or inherit the same planner premise \cite{lewis2020retrieval,barnett2024seven,schick2023toolformer,patil2023gorilla}. Consequently, nominally diverse reviewers often agree for identical erroneous reasons: \textbf{three votes may constitute only one epistemic witness.} This common-mode failure arises from shared causal ancestry rather than residual stochastic correlation.

In physical infrastructure, systems distinguish hosts, racks, power domains, and availability zones because replica count alone does not capture correlated failure exposure \cite{ford2010availability,cidon2013copysets,weil2006crush}. Agentic authorization requires an analogous model of epistemic ancestry. We term these decision-specific common-cause sets \emph{Epistemic Fault Domains} (EFDs). An EFD links a modeled upstream fault to every reviewer it structurally reaches, defined relative to an explicit Epistemic Fault Basis. \textbf{Agent identity is not a failure boundary.}

This work addresses cognitive consensus rather than classical Byzantine fault tolerance (BFT) \cite{lamport1982byzantine,castro1999practical}. While BFT guarantees protocol-level safety under explicit replica fault assumptions, we ask: \emph{when protocol-compliant cognitive reviewers approve an action, how many distinct epistemic roots support a decisive authorizing coalition?} This question builds directly on honest-but-wrong quorums and commit-time agent transaction processing \cite{pbf2026,tct2026,mnemosyne2026,commit_authorization2026}.

Our core metric is the Structural Epistemic Cut $\kappa_E$: the minimum number of modeled fault roots whose exposure covers an authorizing coalition. Under a 3-of-5 rule where one stale telemetry root reaches three reviewers, $\kappa_E=1$ despite five nominal votes. We distinguish structural exposure from realized cognitive failure and unsafe authorization: the Semantic Compromise Cut $\kappa_S$ measures the minimum faults required to produce an actual unsafe commit. Under authorization alignment, closed causal accounting, and complete conservative exposure, $\kappa_E \le \kappa_S$.

This theory yields four operational results: (1)~arbitrarily large quorums can retain $\kappa_E=1$; (2)~threshold resilience is exactly characterized by domain coverage; (3)~discovering unmodeled common ancestry can only decrease or preserve credited resilience; and (4)~appending participants under a fixed threshold cannot increase the cut under compatible extensions. A weak quorum must be reconfigured by replacing coupled evidence paths or adjusting decision rules; voter accumulation at the old threshold is provably insufficient.

Dependency-Aware Quorum Admission (DAQC) enforces prospective selection before execution and evaluates realized provenance $\widehat D^{\mathcal B}$ before admitting a \texttt{COMMIT}. We evaluate the programmed fault-propagation mechanism through closed-form derivations and seeded Monte Carlo simulations, and freeze a 120-task operational benchmark and model endpoint protocol for external evaluation.

This paper makes five primary contributions.
First, we introduce \emph{Epistemic Fault Domains} (EFDs), formalizing common-cause dependency structures in multi-agent authorization relative to an explicit Epistemic Fault Basis.
Second, we define the \emph{Structural Epistemic Cut} $\kappa_E$ alongside the \emph{Semantic Compromise Cut} $\kappa_S$, proving that the structural cut lower-bounds the number of root faults required for semantic compromise under closed causal accounting, conservative exposure, and authorization alignment.
Third, we develop a formal \emph{Quorum-Resilience Theory}, proving cardinality insufficiency ($\kappa_E=1$ for arbitrarily large quorums), characterizing exact threshold coverage, bounding domain fan-out, and establishing participant-addition monotonicity at fixed thresholds.
Fourth, we design the \emph{Dependency-Aware Quorum Controller} (DAQC), which separates prospective quorum selection over planned ancestry from commit-time authorization over realized provenance.
Fifth, we provide a \emph{Layered Evaluation and Benchmark}, combining closed-form baseline derivations, Monte Carlo mechanism checks across diverse topologies, and a frozen 120-task benchmark suite with a standardized endpoint contract for external model validation.
The guiding operational principle across these contributions is: \textbf{count fault-separated epistemic paths, not agent instances.}

The remainder of this paper is organized as follows. Section~\ref{sec:system-model} establishes the multi-agent system and causal exposure model. Section~\ref{sec:fault-domains} formalizes Epistemic Fault Domains and proves the core cut resilience theorems. Section~\ref{sec:quorum-admission} presents the DAQC controller architecture. Section~\ref{sec:evaluation-methodology} validates the fault-propagation mechanics via analytical derivations and simulations. Section~\ref{sec:real-agent-protocol} defines the frozen 120-task benchmark suite and endpoint protocol. Section~\ref{sec:limitations} discusses limitations, Section~\ref{sec:related-work} reviews related work, and Section~\ref{sec:conclusion} concludes.

\section{System Model and Epistemic Dependencies}
\label{sec:system-model}

To formalize dependency-aware authorization, we specify how fault roots propagate to cognitive reviewers. The model represents a completed decision execution wherein cognitive participants evaluate an operational proposition under a monotone quorum rule. It formally separates the actual structural dependency graph and exposure map, their runtime reconstructions, and counterfactual semantic criticality.

\subsection{Participants, Decisions, and Quorums}

Let $\mathcal{A}=\{a_1,\ldots,a_n\}$ denote available cognitive participants (e.g., LLM agents, deterministic verifiers, or humans), and let $\Phi$ be the set of candidate actions. For a proposition $\phi\in\Phi$ evaluated at time $t$, candidate quorum $Q\subseteq\mathcal{A}$ emits judgment profile $\mathbf{y}_Q=\{y_i(\phi,t):a_i\in Q\}$ with $y_i\in\mathcal{Y}=\{\texttt{APPROVE},\texttt{REJECT}\}$. A monotone rule governs authorization:
\begin{equation}
\label{eq:quorum-rule}
  \Gamma(Q,\mathbf{y}_Q)\in\{\texttt{COMMIT},\texttt{ABORT}\}.
\end{equation}
Monotonicity ensures that switching an individual judgment from \texttt{REJECT} to \texttt{APPROVE} cannot transform a commit outcome into an abort.

\subsection{Semantic Validity and Fault Activation}

Let $\omega_t$ denote the authoritative external state. Predicate $\operatorname{ValidJudgment}(\phi,y_i;\omega_t)$ denotes judgment correctness, while $\operatorname{SafeCommit}(\phi;\omega_t)$ denotes compliance with the authoritative safety specification.

\begin{definition}[Participant semantic failure]
Participant $a_i$ fails semantically on $(\phi,t)$ when
\begin{equation}
\label{eq:participant-failure}
  F_i(\phi,t)\triangleq \mathbf{1}\!\left[\neg\operatorname{ValidJudgment}(\phi,y_i;\omega_t)\right].
\end{equation}
\end{definition}

Protocol compliance and semantic validity are orthogonal: an authenticated, protocol-compliant agent can emit an invalid judgment when evaluating an erroneous premise.

\begin{definition}[Quorum unsafe-authorization failure]
A quorum fails when it commits an unsafe action:
\begin{equation}
\label{eq:quorum-failure}
  \begin{aligned}
  F_Q(\phi,t)\triangleq \mathbf{1}\bigl[\,&\Gamma(Q,\mathbf{y}_Q)=\texttt{COMMIT} \\
  &\land\; \neg\operatorname{SafeCommit}(\phi;\omega_t)\bigr].
  \end{aligned}
\end{equation}
\end{definition}

\begin{definition}[Epistemic Fault Basis]
\label{def:fault-basis}
An \emph{Epistemic Fault Basis} $\mathcal{B}_{\phi,t}$ is the finite set of modeled exogenous fault events against which common-cause resilience for $(\phi,t)$ is evaluated. Each $c\in\mathcal{B}_{\phi,t}$ denotes an admissible failure event at the selected threat-model resolution, with defined activation and propagation semantics.
\end{definition}

\noindent\textbf{Epistemic redundancy is relative to an explicit fault basis.} Endpoint, service, upstream-data, availability-zone, and provider bases make distinct resilience claims. We write $\kappa_E^{\mathcal B}(Q,\phi,t,\Gamma)$ when the basis is explicit and omit $\mathcal B$ when fixed by context. A well-formed basis satisfies: (1)~\emph{Failure semantics} (defined fault activation); (2)~\emph{Common-ancestor closure} within the threat boundary; (3)~\emph{Non-duplication} (unique causal identifiers); and (4)~\emph{Explicit threat boundaries}.

Each $c\in\mathcal{B}_{\phi,t}$ has an activation variable $Z_c\in\{0,1\}$ ($Z_c=1$ denotes activation) and admissible realization set $\Delta(c,\phi,t)$. The term \emph{basis} does not assume stochastic independence across $Z_c$. To ensure a finite, well-defined cut, each candidate quorum is completed with one local semantic-fault event $\ell_i$ per $a_i\in Q$:
\begin{equation}
\label{eq:completed-basis}
  \mathcal B^+_{\phi,t}(Q) \triangleq \mathcal B_{\phi,t}\cup\{\ell_i:a_i\in Q\}, \qquad D(\ell_i)=\{a_i\}.
\end{equation}
We extend the activation and realization semantics to the completed basis $\mathcal B^+_{\phi,t}(Q)$: each local root $\ell_i$ denotes an admissible participant-local semantic-fault event with activation variable $Z_{\ell_i}$. For any $C\subseteq\mathcal B^+_{\phi,t}(Q)$, $\Delta(C)$ denotes jointly realizable assignments in which the roots in $C$ are activated and all modeled roots outside $C$ are held nominal. All cut definitions minimize over this completed basis.

\paragraph{Resolution and Common Ancestry.} Consider a unanimous 2-reviewer quorum. At endpoint resolution with paths $c_A\rightarrow\text{API}_A\rightarrow a_1$ and $c_B\rightarrow\text{API}_B\rightarrow a_2$, $D(c_A)=\{a_1\}$ and $D(c_B)=\{a_2\}$, yielding $\kappa_E^{\mathrm{endpoint}}=2$. If both APIs query a shared database ($c_{\mathrm{DB}}\rightarrow\text{DB}\rightarrow\text{API}_{A,B}\rightarrow a_{1,2}$), $D(c_{\mathrm{DB}})=\{a_1,a_2\}$, reducing the database-aware cut to $\kappa_E^{\mathrm{DB}}=1$.

Shared causal ancestry and residual statistical dependence are distinct. Even when $F_i\perp F_j\mid Z_c$, marginal dependence $F_i\not\perp F_j$ can arise because $Z_c$ shifts both conditional failure probabilities:
\begin{equation}
\label{eq:total-covariance}
\begin{aligned}
\operatorname{Cov}(F_i,F_j) ={}& \mathbb{E}\!\left[\operatorname{Cov}(F_i,F_j\mid Z_c)\right] \\
&+ \operatorname{Cov}\!\left(\mathbb{E}[F_i\mid Z_c],\mathbb{E}[F_j\mid Z_c]\right).
\end{aligned}
\end{equation}
The second term captures common-cause dependence. We refer to simultaneous failures from shared activated causes as \emph{common-mode co-failure}.

\subsection{Actual Graph, Runtime Reconstruction, and Semantic Criticality}

Each decision execution induces an actual DAG $G_E^\circ=(V,E^\circ)$ over completed basis roots, evidence artifacts, transformations, and participants. The controller reconstructs runtime graph $\widehat G_E=(\widehat V,\widehat E)$ from provenance logs and attestations.

\begin{definition}[Actual and runtime structural exposure]
For $c\in\mathcal B^+_{\phi,t}(Q)$, actual and runtime exposure maps are:
\begin{align}
\label{eq:actual-runtime-structural-domains}
D_{\phi,t}^{\circ,\mathcal B}(c)&\triangleq\{a_i\in Q:c\leadsto_{G_E^\circ}a_i\},\\
\widehat D_{\phi,t}^{\mathcal B}(c)&\triangleq\{a_i\in Q:c\leadsto_{\widehat G_E}a_i\}.
\end{align}
\end{definition}

Counterfactual semantic criticality is distinct. For $S\subseteq\mathcal B^+_{\phi,t}(Q)\setminus\{c\}$, let $\operatorname{Crit}(a_i,c\mid S,\phi,t) \Longleftrightarrow \exists\delta_S,\delta_c,\delta'_c: F_i^{\delta_S,c\leftarrow\delta_c}(\phi,t) \ne F_i^{\delta_S,c\leftarrow\delta'_c}(\phi,t)$. The semantic counterfactual domain is:
\begin{equation}
\label{eq:critical-domain}
\begin{aligned}
D_{\phi,t}^{\mathrm{crit},\mathcal B}(c) \triangleq\bigl\{a_i\in Q:\;&\exists S\subseteq\mathcal B^+_{\phi,t}(Q)\setminus\{c\},\\
&\operatorname{Crit}(a_i,c\mid S,\phi,t)\bigr\}.
\end{aligned}
\end{equation}
Runtime provenance cannot infer counterfactual criticality; DAQC therefore enforces the structural admission policy over reconstructed structural exposure $\widehat D^{\mathcal B}$.

\begin{definition}[Structural epistemic separation]
Two participants are structurally separated under exposure map $D$ when they share no modeled cause:
\begin{equation}
\label{eq:contextual-independence}
  \begin{aligned}
  &\operatorname{StructSeparated}_D(a_i,a_j\mid\phi,t) \\
  &\quad\Longleftrightarrow \nexists c\in\mathcal{B}^+_{\phi,t}(Q):\{a_i,a_j\}\subseteq D(c).
  \end{aligned}
\end{equation}
\end{definition}

\section{Epistemic Fault Domains and Quorum Resilience}
\label{sec:fault-domains}
\label{sec:quorum-resilience}

This section formalizes structural Epistemic Fault Domains (EFDs), defines structural and semantic resilience cuts, and characterizes threshold-quorum resilience via domain coverage.

\subsection{Exposure Domains and the Active Hypergraph}

An EFD groups participants whose judgments are reachable from a common modeled cause. For any selected structural map $D:\mathcal B^+_{\phi,t}(Q)\rightarrow 2^Q$, write $D_Q(c)=D(c)\cap Q$. The active hypergraph induced by $D$ is:
\begin{equation}
\label{eq:efd-hypergraph}
  H_E^D(Q,\phi,t)\triangleq \left(Q,\{D_Q(c):c\in\mathcal{B}^+_{\phi,t}(Q),D_Q(c)\ne\emptyset\}\right).
\end{equation}
Domains can overlap across telemetry, retrieval, tools, and prompts. EFD membership denotes structural reachability under $D$; it does not imply that every exposed reviewer fails.

\begin{figure}[t]
  \centering
  \begin{tikzpicture}[
    scale=0.88, transform shape, font=\sffamily,
    cause/.style={draw,rounded corners=3pt,font=\scriptsize\sffamily,align=center,
      inner sep=3pt,line width=0.7pt,minimum width=22mm,minimum height=8.5mm},
    cause_crimson/.style={cause,fill=crimsoncol!7,draw=crimsoncol!90,text=crimsoncol!90!black},
    cause_emerald/.style={cause,fill=emeraldcol!7,draw=emeraldcol!85,text=emeraldcol!85!black},
    cause_cobalt/.style={cause,fill=cobaltcol!7,draw=cobaltcol!90,text=cobaltcol!90!black},
    agent/.style={circle,draw=navycol,fill=white,line width=0.8pt,minimum size=7.5mm,
      inner sep=1pt,align=center,text=navycol},
    arr/.style={-{Latex[length=1.8mm,width=1.2mm]},line width=0.7pt},
    hlabel/.style={font=\tiny\bfseries\sffamily,inner sep=1.8pt,rounded corners=2pt}]
    \node[font=\scriptsize\bfseries\sffamily,text=darkslate] at (0,3.4)
      {Selected Fault Basis $\mathcal{B}_{\phi,t}$};
    \node[cause_crimson] (z1) at (-2.5,2.6) {\textbf{$c_1$ Telemetry}\\[-1pt]{\tiny metric stream}};
    \node[cause_emerald] (z2) at (0,2.6) {\textbf{$c_2$ Retrieval}\\[-1pt]{\tiny runbook source}};
    \node[cause_cobalt] (z3) at (2.5,2.6) {\textbf{$c_3$ Tool}\\[-1pt]{\tiny calculator API}};
    \foreach \x/\n in {-2.4/1,-0.8/2,0.8/3,2.4/4}
      \node[agent] (a\n) at (\x,0) {\textbf{\scriptsize $a_\n$}};
    \begin{scope}[on background layer]
      \draw[fill=crimsoncol!7,draw=crimsoncol!90,dashed,rounded corners=6pt]
        (-3.1,-0.5) rectangle (1.4,0.45);
      \draw[fill=emeraldcol!7,draw=emeraldcol!85,dotted,rounded corners=6pt]
        (-1.4,-0.6) rectangle (3.1,0.4);
      \draw[fill=cobaltcol!7,draw=cobaltcol!90,dash pattern=on 3pt off 1.5pt,
        rounded corners=6pt] (0.1,-0.45) rectangle (3.1,0.5);
    \end{scope}
    \foreach \n in {1,2,3}\draw[arr,draw=crimsoncol!85] (z1.south) to[out=270,in=90] (a\n.north);
    \foreach \n in {2,4}\draw[arr,draw=emeraldcol!85] (z2.south) to[out=270,in=90] (a\n.north);
    \foreach \n in {3,4}\draw[arr,draw=cobaltcol!85] (z3.south) to[out=270,in=90] (a\n.north);
    \node[hlabel,fill=crimsoncol!15,text=crimsoncol!90!black] at (-1.9,-0.8)
      {$D(c_1)=\{a_1,a_2,a_3\}$};
    \node[hlabel,fill=emeraldcol!15,text=emeraldcol!85!black] at (0.3,-0.8)
      {$D(c_2)=\{a_2,a_4\}$};
    \node[hlabel,fill=cobaltcol!15,text=cobaltcol!90!black] at (2.2,-0.8)
      {$D(c_3)=\{a_3,a_4\}$};
  \end{tikzpicture}
  \caption{Overlapping EFDs in an active hypergraph. Single cause $c_1$ structurally exposes a decisive 3-of-4 coalition, yielding $\kappa_E=1$.}
  \label{fig:efd-hypergraph}
\end{figure}
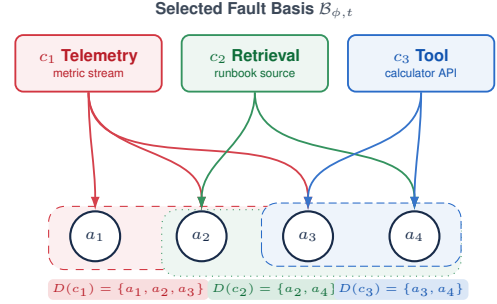

\subsection{Decisive Coalitions and the Structural Cut}

For coalition analysis, identify \texttt{COMMIT} with 1 and \texttt{ABORT} with 0, and write $\Gamma(\mathbf{x})$ for the rule evaluated on approval indicator vector $\mathbf{x}\in\{0,1\}^{|Q|}$ (where $x_i=1 \iff y_i=\texttt{APPROVE}$). Let $\mathbf{1}_W$ denote the indicator vector for coalition $W\subseteq Q$.

\begin{definition}[Minimal decisive coalition]
\label{eq:minimal-decisive-coalitions}
\begin{equation}
\begin{aligned}
\mathcal{W}_{\min}(Q,\Gamma)\triangleq\bigl\{W\subseteq Q:\;& \Gamma(\mathbf{1}_W)=1 \\
&\land\; \forall W'\subsetneq W,\ \Gamma(\mathbf{1}_{W'})=0\bigr\}.
\end{aligned}
\end{equation}
\end{definition}
For an $m$-of-$q$ rule, $\mathcal{W}_{\min}=\{W\subseteq Q:|W|=m\}$; under unanimity, $\mathcal{W}_{\min}=\{Q\}$.

\begin{definition}[Structural Epistemic Cut]
For structural exposure map $D$, the Structural Epistemic Cut is:
\begin{equation}
\label{eq:kappa-definition}
\begin{aligned}
&\kappa_E^{\mathcal B,D}(Q,\phi,t,\Gamma) \\
&\quad\triangleq\min\Bigl\{ |C|:\; C\subseteq\mathcal{B}^+_{\phi,t}(Q), \\
&\qquad\qquad\quad \exists W\in\mathcal{W}_{\min},\; W\subseteq\bigcup_{c\in C}D_Q(c)\Bigr\}.
\end{aligned}
\end{equation}
\end{definition}
Actual-exposure and runtime cuts are $\kappa_E^{\mathcal B,\circ}\triangleq\kappa_E^{\mathcal B,D^{\circ,\mathcal B}}$ and $\widehat\kappa_E^{\mathcal B}\triangleq\kappa_E^{\mathcal B,\widehat D^{\mathcal B}}$. The cut measures the minimum number of modeled roots required to expose a decisive coalition.

\begin{figure}[t]
  \centering
  \begin{tikzpicture}[
    scale=0.85, transform shape,
    box/.style={draw,rounded corners=2pt,font=\scriptsize\sffamily,align=center,fill=white,inner sep=2.5pt},
    dep/.style={box,fill=crimson!12,draw=crimson},agent/.style={box,fill=navy!10,draw=navy},
    arr/.style={-{Latex[length=1.5mm]},line width=0.55pt,color=slate}]
    \node[font=\scriptsize\bfseries\sffamily] at (0,2.6) {$Q_A$: 3-of-5, $|Q_A|=5$, $\kappa_E=1$};
    \node[dep] (ca) at (0,2.0) {telemetry cause $c_s$};
    \foreach \x/\n in {-1.8/1,-0.9/2,0/3,0.9/4,1.8/5}{
      \node[agent] (a\n) at (\x,1.0) {$a_\n$};\draw[arr] (ca)--(a\n);}
    \node[font=\scriptsize\bfseries\sffamily] at (0,-0.05) {$Q_B$: 3-of-3, $|Q_B|=3$, $\kappa_E=3$};
    \foreach \x/\n in {-1.6/1,0/2,1.6/3}{
      \node[dep] (c\n) at (\x,-0.7) {$c_\n$};\node[agent] (b\n) at (\x,-1.5) {$b_\n$};
      \draw[arr] (c\n)--(b\n);}
  \end{tikzpicture}
  \caption{Cardinality and cut measure different properties. $Q_A$ has 5 voters but $\kappa_E=1$; $Q_B$ has 3 voters with separated paths and $\kappa_E=3$.}
  \label{fig:cut-comparison}
\end{figure}
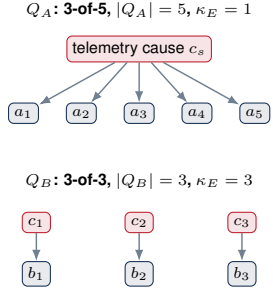

\subsection{Semantic Compromise and Its Relation to Exposure}

\begin{definition}[Semantic Compromise Cut]
For jointly realizable fault assignment $\delta_C\in\Delta(C)$, the Semantic Compromise Cut is:
\begin{equation}
\label{eq:kappa-semantic}
\begin{aligned}
&\kappa_S^{\mathcal B}(Q,\phi,t,\Gamma) \triangleq \min_{\emptyset\ne C\subseteq\mathcal{B}^+_{\phi,t}(Q)} \Bigl\{|C|:\; \exists\delta_C\in\Delta(C), \\
&\quad \Gamma(Q,\mathbf{y}_Q^{\delta_C})=\texttt{COMMIT} \;\land\; \neg\operatorname{SafeCommit}(\phi;\omega_t)\Bigr\}.
\end{aligned}
\end{equation}
\end{definition}
If no admissible assignment produces an unsafe commit, define $\kappa_S^{\mathcal B}=\infty$.

\begin{lemma}[Criticality Requires Structural Exposure]
\label{lem:criticality-exposure}
If $G_E^\circ$ is causally sound, then for every $c\in\mathcal B^+_{\phi,t}(Q)$,
\begin{equation}
\label{eq:criticality-subset-actual}
  D^{\mathrm{crit},\mathcal B}(c) \subseteq D^{\circ,\mathcal B}(c).
\end{equation}
\end{lemma}
\begin{proof}
Membership in $D^{\mathrm{crit},\mathcal B}(c)$ implies that perturbing $c$ alters $a_i$'s failure state. Causal soundness supplies a directed path from $c$ to state consumed by $a_i$, placing $a_i\in D^{\circ,\mathcal B}(c)$.
\end{proof}

We assume: (1)~\emph{Authorization alignment} (approving unsafe $\phi$ is invalid); (2)~\emph{Closed causal accounting} (each invalid approval traces to at least one activated $c\in C$ in $G_E^\circ$); and (3)~\emph{Complete conservative exposure} ($D(c)\supseteq D^{\circ,\mathcal B}(c)$).

\begin{theorem}[Structural Cut Lower-Bounds Semantic Compromise]
\label{thm:structural-semantic}
Under authorization alignment, closed causal accounting, and complete conservative exposure,
\begin{equation}
\kappa_E^{\mathcal B,D}(Q,\phi,t,\Gamma) \le \kappa_S^{\mathcal B}(Q,\phi,t,\Gamma).
\end{equation}
Consequently, $\kappa_E^{\mathcal B,D}\ge k \Longrightarrow \kappa_S^{\mathcal B}\ge k$.
\end{theorem}
\begin{proof}
Let $C$ realize an unsafe commit. Monotonicity implies the approving profile contains some $W\in\mathcal{W}_{\min}$. By authorization alignment, all approvals in $W$ are invalid. By closed causal accounting and complete exposure, $W\subseteq\bigcup_{c\in C}D_Q(c)$, so $|C|\ge\kappa_E^{\mathcal B,D}$. Minimizing over realizing $C$ gives the result.
\end{proof}

\subsection{Provenance Approximation Errors}

Write structural model $\mathcal M_j=(\mathcal B_j^+,D_j)$. Define conservative refinement by $\mathcal M_1\preceq\mathcal M_2 \Longleftrightarrow \mathcal B_1^+\subseteq\mathcal B_2^+ \;\land\; D_1(c)\subseteq D_2(c)\;\forall c\in\mathcal B_1^+$.

\begin{proposition}[Common-Cause Refinement Monotonicity]
\label{prop:refinement-monotonicity}
For fixed $(Q,\phi,t,\Gamma)$,
\begin{equation}
\mathcal M_1\preceq\mathcal M_2 \Longrightarrow \kappa_E^{\mathcal B_2,D_2}(Q) \le\kappa_E^{\mathcal B_1,D_1}(Q).
\end{equation}
\end{proposition}
\begin{proof}
Every valid cause cover in $\mathcal M_1$ remains a cover in $\mathcal M_2$ with widened exposure sets. Additional roots cannot remove covers, so the minimum cannot increase.
\end{proof}

For interpretation only, define $\kappa_E^{\mathcal B,\mathrm{crit}}\triangleq\kappa_E^{\mathcal B,D^{\mathrm{crit},\mathcal B}}$. This quantity is not used for runtime admission. Under conservative runtime overapproximation ($\widehat D^{\mathcal B}(c)\supseteq D^{\circ,\mathcal B}(c)$), $\widehat\kappa_E^{\mathcal B}\le\kappa_E^{\mathcal B,\circ}$. Combining Lemma~\ref{lem:criticality-exposure} yields the domain and structural cut orderings:
\begin{equation}
\label{eq:three-domain-ordering}
\begin{aligned}
D^{\mathrm{crit},\mathcal B}(c) &\subseteq D^{\circ,\mathcal B}(c) \subseteq\widehat D^{\mathcal B}(c),\\
\widehat\kappa_E^{\mathcal B} &\le \kappa_E^{\mathcal B,\circ} \le \kappa_E^{\mathcal B,\mathrm{crit}}.
\end{aligned}
\end{equation}
Combining conservative runtime overapproximation with Theorem~\ref{thm:structural-semantic} under authorization alignment, closed causal accounting, and complete conservative exposure gives
\begin{equation}
\label{eq:semantic-cut-ordering}
\widehat\kappa_E^{\mathcal B} \le \kappa_E^{\mathcal B,\circ} \le \kappa_S^{\mathcal B}.
\end{equation}
We do not assert $\kappa_E^{\mathcal B,\mathrm{crit}}\le\kappa_S^{\mathcal B}$; counterfactual criticality is defined existentially across potentially different conditioning sets $S$, so causes whose criticality domains cover a decisive coalition need not be jointly realizable in a single unsafe execution.

\subsection{Structural Characterizations}

\begin{theorem}[Cardinality Does Not Imply Epistemic Resilience]
\label{thm:cardinality-insufficiency}
For every $q\ge1$ and every non-trivial monotone approval rule $\Gamma_q$, there exists an admissible exposure map with $D_Q(c)=Q$ and $\kappa_E(Q)=1$. Thus $|Q|\rightarrow\infty \not\Longrightarrow \kappa_E(Q)\rightarrow\infty$.
\end{theorem}
\begin{proof}
$\mathcal{W}_{\min}$ is nonempty and every one of its coalitions is covered by $\{c\}$, so $\kappa_E\le1$. Every decisive coalition is nonempty, so $\kappa_E\ge1$.
\end{proof}

For an $m$-of-$q$ quorum, define coverage function $h_Q(s)\triangleq\max_{|C|\le s}|\bigcup_{c\in C}D_Q(c)|$ with $h_Q(0)=0$.

\begin{theorem}[Threshold-Quorum Coverage Characterization]
\label{thm:coverage-characterization}
For an $m$-of-$q$ threshold quorum:
\begin{align}
  \kappa_E(Q)&=\min\{s:h_Q(s)\ge m\},\\
  \kappa_E(Q)&\ge k \iff h_Q(k-1)<m.
\end{align}
\end{theorem}
\begin{proof}
A cause set of size $\le s$ covers an $m$-member decisive coalition iff $h_Q(s)\ge m$.
\end{proof}

\begin{proposition}[Participant-Addition Monotonicity at Fixed Threshold]
\label{prop:participant-addition}
Let $Q\subseteq Q'$ share $m$-approval threshold with compatible structural models ($\mathcal B^+_{\phi,t}(Q)\subseteq\mathcal B^+_{\phi,t}(Q')$ and $D_{Q'}(c)\cap Q=D_Q(c)$). Then $\kappa_E(Q';m)\le\kappa_E(Q;m)$.
\end{proposition}
\begin{proof}
Every $m$-member decisive coalition in $Q$ remains decisive in $Q'$, and existing covers remain feasible in $Q'$.
\end{proof}

\begin{corollary}[Fixed-Threshold Expansion Cannot Repair Cut Deficit]
\label{cor:fixed-threshold-deficit}
If $\kappa_E(Q;m)<k_{\min}$, then for all compatible extensions $Q'\supseteq Q$, $\kappa_E(Q';m)<k_{\min}$.
\end{corollary}

\begin{proposition}[Domain Fan-Out Bounds]
\label{prop:fanout-lower-bound}
Let $m_{\min}=\min_{W\in\mathcal{W}_{\min}}|W|$ and $r=\max_c|D_Q(c)|$. Then:
\begin{equation}
\label{eq:fanout-bounds}
  \left\lceil\frac{m_{\min}}{r}\right\rceil \le \kappa_E(Q) \le m_{\min}.
\end{equation}
\end{proposition}
\begin{proof}
$s$ causes cover at most $sr$ members, requiring $s\ge\lceil m_{\min}/r\rceil$. Singleton local causes bound $\kappa_E\le m_{\min}$.
\end{proof}

\section{Dependency-Aware Quorum Control}
\label{sec:quorum-admission}
\label{sec:path-diversity}

\subsection{Nominal Diversity vs.\ Epistemic Path Diversity}

When mitigating cognitive failures, architectures frequently deploy multi-model quorums $Q=\{a_1,a_2,a_3\}$ across model families $(M_1, M_2, M_3)$. We term differences restricted to participant interfaces \emph{nominal diversity} (e.g., model checkpoints, prompt templates, system personas, random seeds).

Nominal diversity provides zero structural protection against shared epistemic causes: querying three distinct models over one corrupted telemetry stream $z_{\mathrm{tele}}$ yields $\kappa_E=1$. Three structural patterns illustrate this failure:
(1)~\emph{Replicated reasoning with shared computation}: $(M_1, M_1, M_1)$ sharing a prompt trap or tokenizer bug ($D(c_{\mathrm{model}})=\{a_1,a_2,a_3\}$);
(2)~\emph{Heterogeneous models with shared evidence}: $(M_1,E),(M_2,E),(M_3,E)$ sharing corrupted telemetry $E$ ($D(c_E)=\{a_1,a_2,a_3\}$, $\kappa_E=1$);
(3)~\emph{Distinct interfaces with shared ancestry}: separate caches descending from a single corrupted database root ($c_0\rightarrow E_j$).

In contrast, a quorum achieves \emph{epistemic path diversity} when decisive judgments are supported by separated modeled fault roots in $G_E^\circ$. Relative to the evidence-focused basis used in Section~\ref{sec:evaluation} ($\mathcal B_{\mathrm{eval}}^+$), a single-model quorum over separated evidence roots achieves $\kappa_E=2$ under a 2-of-3 rule ($Q_3$); this does not claim resilience to a shared model-family fault outside that basis. In contrast, a multi-model quorum over shared evidence remains $\kappa_E=1$. Achieving resilience requires corroborating operational state across distinct modeled evidence roots (e.g., VPC network flows, service heartbeats, and database transaction logs).

\subsection{Selection and Admission Are Distinct Controls}

\emph{Dependency-aware selection} prospectively chooses reviewers using \emph{planned structural exposure} $D^{\mathrm{plan},\mathcal B}_{\phi,t}$, computing prospective cut $\kappa_E^{\mathcal B,\mathrm{plan}}(Q)\triangleq\kappa_E^{\mathcal B,D^{\mathrm{plan},\mathcal B}}(Q)$. In contrast, \emph{dependency-aware admission} evaluates reconstructed exposure $\widehat D^{\mathcal B}$ from realized provenance after judgments are collected. Let $d\triangleq\Gamma(Q,\mathbf y_Q)\in\{\texttt{COMMIT},\texttt{ABORT}\}$. The structural admission gate enforces:
\begin{equation}
\label{eq:admission-invariant}
  \begin{aligned}
  \operatorname{StructAdmit}(Q,\phi,t,\Gamma) \triangleq \mathbf{1}\bigl[\,&|Q| \ge q_{\min} \\
  &\land\; \widehat\kappa_E^{\mathcal B}(Q,\phi,t,\Gamma) \ge k_{\min}\,\bigr].
  \end{aligned}
\end{equation}
For an $m$-of-$q$ rule, feasibility requires $1\le k_{\min}\le m\le q$. When $d=\texttt{ABORT}$, the system fails closed immediately. When $d=\texttt{COMMIT}$, mutation is authorized only if $\operatorname{StructAdmit}=1$; otherwise, mutation is denied, triggering reconfiguration or escalation.

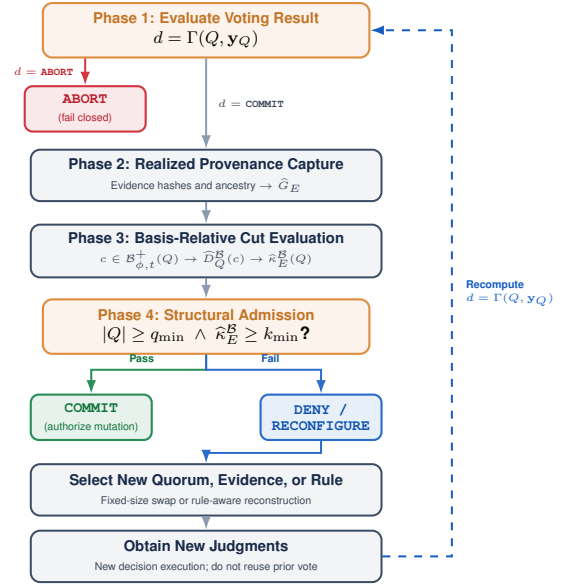
\begin{figure}[t]
  \centering
  \begin{tikzpicture}[
    scale=0.80, transform shape,
    font=\sffamily,
    phase/.style={
      draw=navycol!85, fill=navycol!5, rounded corners=3.5pt, line width=0.75pt,
      font=\scriptsize\sffamily, align=center, inner sep=3.5pt,
      minimum width=58mm, minimum height=8.5mm
    },
    gate/.style={
      draw=ambercol!90, fill=ambercol!9, rounded corners=3.5pt, line width=0.8pt,
      font=\scriptsize\sffamily, align=center, inner sep=3.5pt,
      minimum width=54mm, minimum height=9mm
    },
    outcome/.style={
      rounded corners=3pt, line width=0.8pt,
      font=\scriptsize\sffamily, align=center, inner sep=3pt,
      minimum width=20mm, minimum height=7.5mm
    },
    admit/.style={outcome, draw=emeraldcol!90, fill=emeraldcol!10, text=emeraldcol!90!black},
    reject/.style={outcome, draw=crimsoncol!90, fill=crimsoncol!10, text=crimsoncol!90!black},
    control/.style={outcome, draw=cobaltcol!90, fill=cobaltcol!10, text=cobaltcol!90!black},
    arr/.style={-{Latex[length=1.8mm, width=1.3mm]}, line width=0.75pt, draw=linegray},
    fb_arr/.style={-{Latex[length=1.8mm, width=1.3mm]}, line width=0.75pt, dashed, draw=cobaltcol}
  ]

    \node[gate] (vote) at (0, 5.2) {
      \textbf{\color{ambercol!90!black}Phase 1: Evaluate Voting Result}\\[1pt]
      {\footnotesize\bfseries $d=\Gamma(Q,\mathbf{y}_Q)$}
    };

    \node[reject] (abort) at (-2.0, 3.9) {\textbf{\texttt{ABORT}}\\[0.5pt]{\tiny (fail closed)}};

    \node[phase] (p1) at (0, 2.8) {
      \textbf{\color{navycol}Phase 2: Realized Provenance Capture}\\[1pt]
      {\tiny\color{darkslate}Evidence hashes and ancestry $\rightarrow \widehat{G}_E$}
    };

    \node[phase] (p2) at (0, 1.55) {
      \textbf{\color{navycol}Phase 3: Basis-Relative Cut Evaluation}\\[1pt]
      {\tiny\color{darkslate}$c \in \mathcal{B}^+_{\phi,t}(Q) \rightarrow \widehat{D}_Q^{\mathcal{B}}(c) \rightarrow \widehat{\kappa}_E^{\mathcal{B}}(Q)$}
    };

    \node[gate] (gate) at (0, 0.3) {
      \textbf{\color{ambercol!90!black}Phase 4: Structural Admission}\\[1pt]
      {\footnotesize\bfseries $|Q| \ge q_{\min} \;\land\; \widehat{\kappa}_E^{\mathcal{B}} \ge k_{\min}$?}
    };

    \node[admit]  (admit)  at (-1.9, -1.2) {\textbf{\texttt{COMMIT}}\\[0.5pt]{\tiny (authorize mutation)}};
    \node[control] (control) at (1.9, -1.2) {\textbf{\texttt{DENY /}}\\[-0.5pt]\textbf{\texttt{RECONFIGURE}}};

    \node[phase] (select) at (0, -2.4) {
      \textbf{\color{navycol}Select New Quorum, Evidence, or Rule}\\[1pt]
      {\tiny\color{darkslate}Fixed-size swap or rule-aware reconstruction}
    };
    \node[phase] (execute) at (0, -3.5) {
      \textbf{\color{navycol}Obtain New Judgments}\\[1pt]
      {\tiny\color{darkslate}New decision execution; do not reuse prior vote}
    };

    \draw[arr, draw=crimsoncol!90] (vote.south -| abort.north) -- (abort.north)
      node[midway, left=2pt, font=\tiny\bfseries\sffamily, text=crimsoncol!90!black] {$d = \texttt{ABORT}$};

    \draw[arr] (vote.south) -- (p1.north)
      node[midway, right=2pt, font=\tiny\bfseries\sffamily, text=darkslate] {$d = \texttt{COMMIT}$};

    \draw[arr] (p1) -- (p2);
    \draw[arr] (p2) -- (gate);

    \draw[arr, draw=emeraldcol!90] (gate.south) -- ++(0, -0.25) -| (admit.north)
      node[pos=0.28, above, font=\tiny\bfseries\sffamily, text=emeraldcol!90!black] {Pass};

    \draw[arr, draw=cobaltcol!90] (gate.south) -- ++(0, -0.25) -| (control.north)
      node[pos=0.28, above, font=\tiny\bfseries\sffamily, text=cobaltcol!90!black] {Fail};

    \draw[arr, draw=cobaltcol!90] (control.south) -- ++(0, -0.2) -| (select.north);
    \draw[arr] (select) -- (execute);

    \draw[fb_arr] (execute.east) -- ++(1.15, 0)
      coordinate (fb_bot) -- (fb_bot |- vote.east)
      node[midway, right=3pt, font=\tiny\bfseries\sffamily, text=cobaltcol, align=left] {Recompute\\$d=\Gamma(Q,\mathbf y_Q)$}
      -- (vote.east);

  \end{tikzpicture}
  \caption{Reference DAQC mutation-admission architecture. An \texttt{ABORT} vote fails closed immediately. A \texttt{COMMIT} vote requires verified cut $\widehat\kappa_E^{\mathcal B}\ge k_{\min}$; failed checks deny mutation and require full reconfiguration.}
  \label{fig:admission-procedure}
\end{figure}

Reconfiguration is an independent selection--execution cycle. It replaces coupled reviewers ($Q'=(Q\setminus R)\cup N$), adjusts threshold $(q,m)\mapsto(q',m')$, or acquires newly separated evidence, requiring fresh judgments and recomputation of $d=\Gamma(Q,\mathbf y_Q)$. Appending voters while retaining threshold $m$ cannot repair a structural cut deficit (Corollary~\ref{cor:fixed-threshold-deficit}).

For selection, the constrained optimization over candidates $\mathcal{A}_{\mathrm{cand}}$ with target $(q^\star, m^\star)$ is:
\begin{equation}
\label{eq:constrained-selection}
  \begin{aligned}
  \min_{Q\subseteq\mathcal{A}_{\mathrm{cand}}} \quad & C(Q) + \lambda L(Q) \\
  \text{s.t.} \quad & |Q|=q^\star, \quad \Gamma=\Gamma_{m^\star,q^\star},\\
  &\kappa_E^{\mathcal B,\mathrm{plan}}(Q;m^\star)\ge k_{\min}, \quad Q\models\Pi_{\mathrm{policy}},
  \end{aligned}
\end{equation}
where $\lambda\ge0$ is a policy-selected latency weight. Algorithm~\ref{alg:daqc-swap} implements a fixed-cardinality swap heuristic ($Q'=(Q\setminus\{a\})\cup\{b\}$) accepting lexicographic improvements in cut and cost.

\subsection{Reference Architecture and Conservative Enforcement}

DAQC operates across four phases (Figure~\ref{fig:admission-procedure}): (1)~evaluate vote $d$; (2)~capture realized DAG $\widehat G_E$; (3)~construct active domains and compute $\widehat\kappa_E^{\mathcal B}$; and (4)~admit mutation iff Equation~\eqref{eq:admission-invariant} holds. Unknown dependencies are conservatively classified as shared ($D(c)\supseteq D^{\circ,\mathcal B}(c)$), ensuring $\widehat\kappa_E^{\mathcal B}\le\kappa_E^{\mathcal B,\circ}$ to prevent overstating resilience. Appendix~\ref{sec:appendix-daqc} provides complete algorithms.

\section{Controlled Monte Carlo Mechanism Validation}
\label{sec:evaluation-methodology}
\label{sec:evaluation}

We use seeded Monte Carlo experiments to check that the implementation reproduces the programmed structural mechanisms and closed-form analytical expectations. Three \emph{Mechanism Validation Questions} organize this evaluation:
\emph{MV1---common-cause propagation:} does the implementation reproduce theoretical quorum-failure rates when one root exposes multiple reviewers?
\emph{MV2---shared vs.\ separated exposure:} does it reproduce the analytical contrast between shared-root and separated-root configurations?
\emph{MV3---cardinality and cut:} does it distinguish configurations with different structural cuts despite equal or increasing cardinality?

\subsection{Experimental Setup and Benchmark Suite}
\label{sec:eval-setup}

The synthetic simulator generates 1,000 replayable scenario records (seed 42): 500 infrastructure-mutation and 500 document-grounded policy records (550 safe, 450 unsafe overall). These records serve as standardized scheduling units across fault classes, fan-outs, and thresholds. For each trial, the harness builds exact exposure map $D^{\circ,\mathcal B}$ and computes $\kappa_E^{\mathcal B,\circ}$ by enumerating minimal decisive coalitions.

For the Q1--Q4 path-diversity aggregate, the selected environmental fault basis is $\mathcal B_{\mathrm{eval}}=\mathcal B_{\mathrm{obs}}\cup\mathcal B_{\mathrm{retrieval}}\cup\mathcal B_{\mathrm{tool}}\cup\mathcal B_{\mathrm{derivation}}$, completed with reviewer-local roots as $\mathcal B_{\mathrm{eval}}^+(Q)=\mathcal B_{\mathrm{eval}}\cup\{\ell_i:a_i\in Q\}$. The separate common-cause mechanism study in Table~\ref{tab:eval-amplification} also includes a computation/prompt-proxy cause class; that class is not part of the $\mathcal B_{\mathrm{eval}}^+(Q)$-relative Q1--Q4 cut comparison. Controlled fault injection activates one cause ($Z_c=1$) per faulted execution across fan-outs $|D_Q(c)|\in\{1,2,3,5,7\}$ and derivation depths $d\in\{1,2,3\}$. Conditional on cause activation, each reviewer fails independently with the configured transmission probability. The reported cuts characterize only the roots in $\mathcal B_{\mathrm{eval}}^+(Q)$; a common model-family root is excluded from this aggregate and would expose Q3 if added.

\begin{align}
\Delta_c^Q &= P(F_Q=1\mid\operatorname{do}(Z_c=1)) \notag\\
&\quad - P(F_Q=1\mid\operatorname{do}(Z_c=0)),\\
RR_c^Q &= \frac{P(F_Q=1\mid\operatorname{do}(Z_c=1))}{P(F_Q=1\mid\operatorname{do}(Z_c=0))}.
\end{align}

\subsection{Analytical Failure-Probability Baseline}
\label{sec:analytic-baseline}

Conditional on cause vector $Z$, an $m$-of-$q$ quorum with independent reviewer probabilities $p_i(Z)$ has exact Poisson-binomial tail:
\begin{equation}
\label{eq:poisson-binomial-quorum}
P(F_Q=1\mid Z)=\sum_{\substack{A\subseteq Q\\|A|\ge m}}\prod_{i\in A}p_i(Z)\prod_{j\notin A}\bigl(1-p_j(Z)\bigr),
\end{equation}
which reduces to $P(F_Q=1\mid Z)=\sum_{r=m}^{q}\binom{q}{r}p^r(1-p)^{q-r}$ for homogeneous $p$. For Q1 under 2-of-3, analytical shared-fault SFRs for observation, retrieval, tool, derivation, and computation are .9896, .9664, .9772, .9467, and .8324, respectively; the analytical expectations match the Monte Carlo outcomes within approximately 1.1 percentage points across the reported conditions.

\begin{table*}[t]
  \centering
  \scriptsize
  \caption{Common-cause activation in a 2-of-3 quorum over 450 unsafe base tasks per class. $p_1,p_2,P_{12}$ are conditional on $\operatorname{do}(Z_c=1)$; intervals are Wilson 95\% intervals for faulted quorum SFR. $\Delta^Q_c$ and $RR^Q_c$ measure intervention risk differences.}
  \label{tab:eval-amplification}
  \begin{tabular*}{\textwidth}{@{\extracolsep{\fill}}l c c c c c c c c@{}}
    \toprule
    Fault class & $p_1$ & $p_2$ & $P_{12}$ & SFR$_0$ & SFR$_1$ & $\Delta^Q_c$ & $RR^Q_c$ & 95\% CI \\
    \midrule
    Observation & .927 & .958 & .884 & .004 & .996 & .991 & 224.0 & [.984,.999] \\
    Retrieval & .873 & .873 & .762 & .007 & .967 & .960 & 145.0 & [.946,.980] \\
    Tool & .907 & .909 & .824 & .004 & .978 & .973 & 220.0 & [.960,.988] \\
    Derivation & .864 & .862 & .749 & .013 & .940 & .927 & 70.5 & [.914,.958] \\
    Computation & .751 & .707 & .520 & .002 & .822 & .820 & 370.0 & [.784,.855] \\
    \bottomrule
  \end{tabular*}
\end{table*}

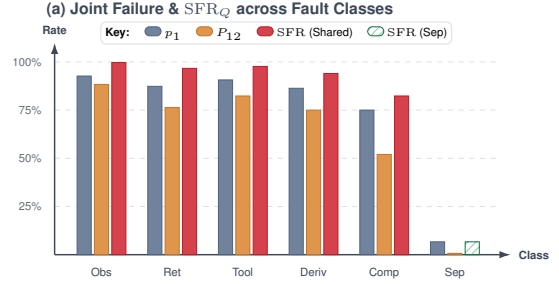
\begin{figure}[t]
  \centering
  \begin{tikzpicture}[
    scale=0.85, transform shape,
    font=\sffamily,
    bar_p1/.style={fill=slatecol!85, draw=slatecol!90!black, line width=0.3pt},
    bar_p12/.style={fill=ambercol!85, draw=ambercol!90!black, line width=0.3pt},
    bar_sfr_shared/.style={fill=crimsoncol!85, draw=crimsoncol!90!black, line width=0.3pt},
    bar_sfr_sep/.style={fill=emeraldcol!85, draw=emeraldcol!90!black, line width=0.3pt, pattern=north east lines, pattern color=emeraldcol!40},
    axis_line/.style={-{Latex[length=1.5mm, width=1.0mm]}, line width=0.55pt, color=darkslate},
    grid_line/.style={line width=0.3pt, color=bordergray, dash pattern=on 2pt off 2pt}
  ]
    \node[font=\scriptsize\bfseries\sffamily, anchor=west, text=darkslate] at (-0.25, 3.8)
      {(a) Joint Failure \& $\operatorname{SFR}_Q$ across Fault Classes};

    \foreach \y/\label in {0.75/25\%, 1.5/50\%, 2.25/75\%, 3.0/100\%} {
      \draw[grid_line] (0,\y) -- (6.9,\y);
      \node[anchor=east, font=\tiny\sffamily, text=darkslate!80] at (-0.08,\y) {\label};
    }

    \draw[axis_line] (0,0) -- (7.1,0) node[right, font=\tiny\bfseries\sffamily, text=darkslate] {Class};
    \draw[axis_line] (0,0) -- (0,3.3) node[above, font=\tiny\bfseries\sffamily, text=darkslate] {Rate};

    \draw[bar_p1]         (0.35,0) rectangle (0.57,2.78);
    \draw[bar_p12]        (0.62,0) rectangle (0.84,2.65);
    \draw[bar_sfr_shared] (0.89,0) rectangle (1.11,2.99);
    \node[font=\tiny\sffamily, text=darkslate, anchor=north] at (0.73, -0.08) {Obs};

    \draw[bar_p1]         (1.45,0) rectangle (1.67,2.62);
    \draw[bar_p12]        (1.72,0) rectangle (1.94,2.29);
    \draw[bar_sfr_shared] (1.99,0) rectangle (2.21,2.90);
    \node[font=\tiny\sffamily, text=darkslate, anchor=north] at (1.83, -0.08) {Ret};

    \draw[bar_p1]         (2.55,0) rectangle (2.77,2.72);
    \draw[bar_p12]        (2.82,0) rectangle (3.04,2.47);
    \draw[bar_sfr_shared] (3.09,0) rectangle (3.31,2.93);
    \node[font=\tiny\sffamily, text=darkslate, anchor=north] at (2.93, -0.08) {Tool};

    \draw[bar_p1]         (3.65,0) rectangle (3.87,2.59);
    \draw[bar_p12]        (3.92,0) rectangle (4.14,2.25);
    \draw[bar_sfr_shared] (4.19,0) rectangle (4.41,2.82);
    \node[font=\tiny\sffamily, text=darkslate, anchor=north] at (4.03, -0.08) {Deriv};

    \draw[bar_p1]         (4.75,0) rectangle (4.97,2.25);
    \draw[bar_p12]        (5.02,0) rectangle (5.24,1.56);
    \draw[bar_sfr_shared] (5.29,0) rectangle (5.51,2.47);
    \node[font=\tiny\sffamily, text=darkslate, anchor=north] at (5.13, -0.08) {Comp};

    \draw[bar_p1]         (5.85,0) rectangle (6.07,0.20);
    \draw[bar_p12]        (6.12,0) rectangle (6.34,0.02);
    \draw[bar_sfr_sep]    (6.39,0) rectangle (6.61,0.20);
    \node[font=\tiny\sffamily, text=darkslate, anchor=north] at (6.23, -0.08) {Sep};

    \node[font=\tiny\sffamily, fill=white, inner sep=2pt, draw=bordergray, rounded corners=2pt]
      at (3.45, 3.45) {
      \textbf{Key:} \quad
      \tikz{\draw[bar_p1] (0,0) rectangle (0.25,0.14);} $p_1$ \quad
      \tikz{\draw[bar_p12] (0,0) rectangle (0.25,0.14);} $P_{12}$ \quad
      \tikz{\draw[bar_sfr_shared] (0,0) rectangle (0.25,0.14);} $\operatorname{SFR}$ (Shared) \quad
      \tikz{\draw[bar_sfr_sep] (0,0) rectangle (0.25,0.14);} $\operatorname{SFR}$ (Sep)
    };
  \end{tikzpicture}
  \caption{Common-cause activation and propagation in simulation. Reviewer marginal, pairwise joint, and quorum failure rates conditional on activated shared causes vs.\ separated evidence.}
  \label{fig:eval-propagation}
\end{figure}

\subsection{Common-Cause Activation and Co-Failure}
\label{sec:eval-propagation}

Table~\ref{tab:eval-amplification} and Figure~\ref{fig:eval-propagation} present conditional failure probabilities under shared exposure ($|D_Q(c)|=3$). Shared fault activation increases quorum SFR from near zero ($<0.01$) to 0.822--0.996 across all fault classes, showing that the implementation reproduces the expected common-mode quorum failure under shared exposure without requiring excess residual correlation. Multi-hop derivations evaluate depths 1--3, yielding SFRs of 95.6\%, 87.8\%, and 76.7\%, respectively, matching analytical attenuation expectations.

\begin{table*}[t]
  \centering
  \scriptsize
  \caption{Synthetic model-label diversity versus path diversity across 8,400 reviewer calls per configuration. Displayed cuts are relative to $\mathcal B_{\mathrm{eval}}^+(Q)$.}
  \label{tab:eval-diversity}
  \begin{tabular*}{\textwidth}{@{\extracolsep{\fill}}c l c c c c c c@{}}
    \toprule
    Rule & Config & $\kappa_E^{\mathcal B_{\mathrm{eval}},\circ}$ & $\operatorname{SFR}_{\mathrm{clean}}$ & $\operatorname{SFR}_{\mathrm{fault}}$ (95\% CI) & $\Delta\operatorname{SFR}$ & $\operatorname{RSAR}$ & Disagree \\
    \midrule
    2-of-3 & Q1 (Replicated, Shared) & 1 & 0.002 & 0.973 [0.965, 0.980] & +0.971 & 0.007 & 25.8\% \\
    2-of-3 & Q2 (Model-Diverse, Shared) & 1 & 0.000 & 0.968 [0.959, 0.975] & +0.968 & 0.004 & 26.1\% \\
    2-of-3 & Q3 (Same Model, Separated) & 2 & 0.007 & 0.067 [0.057, 0.080] & +0.061 & 0.007 & 90.8\% \\
    2-of-3 & Q4 (Model- \& Path-Diverse) & 2 & 0.007 & 0.067 [0.056, 0.079] & +0.060 & 0.011 & 91.7\% \\
    \midrule
    3-of-3 & Q1-3of3 (Replicated, Shared) & 1 & 0.000 & 0.744 [0.724, 0.764] & +0.744 & 0.120 & 25.4\% \\
    3-of-3 & Q2-3of3 (Model-Diverse, Shared) & 1 & 0.000 & 0.727 [0.706, 0.747] & +0.727 & 0.102 & 27.2\% \\
    3-of-3 & Q3-3of3 (Same Model, Separated) & 3 & 0.000 & 0.001 [0.000, 0.004] & +0.001 & 0.115 & 90.6\% \\
    3-of-3 & Q4-3of3 (Model- \& Path-Diverse) & 3 & 0.000 & 0.001 [0.000, 0.003] & +0.001 & 0.118 & 91.4\% \\
    \bottomrule
  \end{tabular*}
\end{table*}

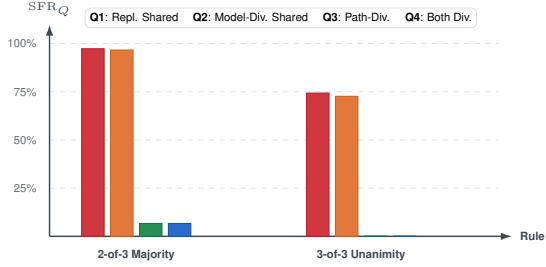
\begin{figure}[t]
  \centering
  \begin{tikzpicture}[
    scale=0.85, transform shape,
    font=\sffamily,
    bar_q1/.style={fill=crimsoncol!90, draw=crimsoncol!90!black, line width=0.3pt},
    bar_q2/.style={fill=orangecol!90, draw=orangecol!90!black, line width=0.3pt},
    bar_q3/.style={fill=emeraldcol!90, draw=emeraldcol!90!black, line width=0.3pt},
    bar_q4/.style={fill=cobaltcol!90, draw=cobaltcol!90!black, line width=0.3pt},
    axis_line/.style={-{Latex[length=1.5mm, width=1.0mm]}, line width=0.55pt, color=darkslate},
    grid_line/.style={line width=0.3pt, color=bordergray, dash pattern=on 2pt off 2pt}
  ]
    \foreach \y/\label in {0.75/25\%, 1.5/50\%, 2.25/75\%, 3.0/100\%} {
      \draw[grid_line] (0,\y) -- (7.0,\y);
      \node[anchor=east, font=\tiny\sffamily, text=darkslate!80] at (-0.08,\y) {\label};
    }
    \draw[axis_line] (0,0) -- (7.2,0) node[right, font=\tiny\bfseries\sffamily, text=darkslate] {Rule};
    \draw[axis_line] (0,0) -- (0,3.3) node[above, font=\tiny\bfseries\sffamily, text=darkslate] {$\operatorname{SFR}_Q$};

    \draw[bar_q1] (0.5,0) rectangle (0.85,2.92);
    \draw[bar_q2] (0.95,0) rectangle (1.30,2.90);
    \draw[bar_q3] (1.40,0) rectangle (1.75,0.20);
    \draw[bar_q4] (1.85,0) rectangle (2.20,0.20);
    \node[font=\tiny\bfseries\sffamily, text=darkslate, anchor=north] at (1.35, -0.08) {2-of-3 Majority};

    \draw[bar_q1] (4.0,0) rectangle (4.35,2.23);
    \draw[bar_q2] (4.45,0) rectangle (4.80,2.18);
    \draw[bar_q3] (4.90,0) rectangle (5.25,0.003);
    \draw[bar_q4] (5.35,0) rectangle (5.70,0.002);
    \node[font=\tiny\bfseries\sffamily, text=darkslate, anchor=north] at (4.85, -0.08) {3-of-3 Unanimity};

    \node[font=\tiny\sffamily, fill=white, inner sep=2pt, draw=bordergray, rounded corners=2pt]
      at (3.6, 3.4) {
      \textbf{Q1}: Repl. Shared \quad
      \textbf{Q2}: Model-Div. Shared \quad
      \textbf{Q3}: Path-Div. \quad
      \textbf{Q4}: Both Div.
    };
  \end{tikzpicture}
  \caption{Model diversity vs.\ path diversity in simulation. Path separation in Q3 yields $14.5\times$ lower observed SFR under 2-of-3 and $670\times$ lower observed SFR under 3-of-3.}
  \label{fig:eval-diversity}
\end{figure}

\subsection{Nominal Model Diversity vs.\ Epistemic Path Diversity}
\label{sec:eval-diversity}

Table~\ref{tab:eval-diversity} and Figure~\ref{fig:eval-diversity} compare Q1--Q4. In this programmed simulation, changing only model labels provides negligible protection against the shared-evidence fault: under a 2-of-3 rule, Q1 records 97.3\% unsafe commits and Q2 records 96.8\% ($\kappa_E^{\mathcal B_{\mathrm{eval}},\circ}=1$). In contrast, path separation in Q3 raises $\kappa_E^{\mathcal B_{\mathrm{eval}},\circ}$ from 1 to 2 and drops faulted SFR to 6.72\% ($14.5\times$ lower observed SFR). Under unanimity (3-of-3, $\kappa_E=3$), Q3 and Q4 eliminate almost all failures ($0.11\%$ and $0.06\%$ SFR, respectively). Path separation also increases inter-reviewer disagreement from 25.8\% to 90.8\%, surfacing clear runtime signals of upstream faultiness.

\begin{table*}[t]
  \centering
  \scriptsize
  \setlength{\tabcolsep}{3pt}
  \caption{Mechanical structural-exposure audit for the scaling study (450 unsafe tasks per row). $\kappa_E^{\mathcal B,\circ}$ is computed directly from threshold $m$ and fan-out $r$.}
  \label{tab:eval-scaling}
  \begin{tabular*}{\textwidth}{@{\extracolsep{\fill}}l c c l c c c c c c@{}}
    \toprule
    Topology & $q$ & $m$ & Injected exposure & $m_{\min}$ & $\kappa_E^{\mathcal B,\circ}$ & Fan-out & Sim. SFR & Analytic & Error \\
    \midrule
    Single & 1 & 1 & $\{a_1\}$ & 1 & 1 & 1 & .9622 & .9400 & +.0222 \\
    Q-3-Shared & 3 & 2 & $\{a_1,a_2,a_3\}$ & 2 & 1 & 3 & .9956 & .9896 & +.0059 \\
    Q-3-Sep (fan-out 2) & 3 & 2 & $\{a_1,a_2\}$ & 2 & \textbf{1} & 2 & .8667 & .8879 & -.0212 \\
    Q-5-Shared & 5 & 3 & $\{a_1,\ldots,a_5\}$ & 3 & 1 & 5 & 1.0000 & .9980 & +.0020 \\
    Q-5-Sep2 & 5 & 3 & $\{a_1,a_2\}$ & 3 & 2 & 2 & .1156 & .1040 & +.0115 \\
    Q-5-Sep3 & 5 & 3 & $\{a_1\}$ & 3 & 3 & 1 & .0022 & .0092 & -.0070 \\
    Q-7-Shared & 7 & 4 & $\{a_1,\ldots,a_7\}$ & 4 & 1 & 7 & 1.0000 & .9996 & +.0004 \\
    Q-7-Sep & 7 & 4 & $\{a_1\}$ & 4 & 4 & 1 & .0022 & .0011 & +.0011 \\
    \bottomrule
  \end{tabular*}
\end{table*}

\begin{figure}[t]
  \centering
  \begin{tikzpicture}[
    scale=0.82, transform shape,
    font=\sffamily,
    axis_line/.style={-{Latex[length=1.5mm, width=1.0mm]}, line width=0.55pt, color=darkslate},
    grid_line/.style={line width=0.3pt, color=bordergray, dash pattern=on 2pt off 2pt}
  ]
    \fill[emeraldcol!8, rounded corners=2pt] (0.8, 0) rectangle (6.8, 0.22);
    \node[font=\tiny\bfseries\sffamily, text=emeraldcol!90!black, anchor=west] at (1.0, 0.38) {Safe Zone ($\operatorname{SFR} \le 1\%$)};

    \foreach \y/\label in {0.75/25\%, 1.5/50\%, 2.25/75\%, 3.0/100\%} {
      \draw[grid_line] (0.8,\y) -- (6.8,\y);
      \node[anchor=east, font=\tiny\sffamily, text=darkslate!80] at (0.75,\y) {\label};
    }

    \draw[axis_line] (0.8,0) -- (7.0,0) node[right, font=\tiny\bfseries\sffamily, text=darkslate] {$|Q|$};
    \draw[axis_line] (0.8,0) -- (0.8,3.3) node[above, font=\tiny\bfseries\sffamily, text=darkslate] {$\operatorname{SFR}_Q$};

    \foreach \x/\label in {1.4/1, 2.9/3, 4.6/5, 6.3/7} {
      \draw[line width=0.5pt, color=darkslate] (\x, 0) -- (\x, -0.06);
      \node[anchor=north, font=\tiny\bfseries\sffamily, text=darkslate] at (\x, -0.08) {\label};
    }

    \draw[line width=1.1pt, color=crimsoncol] (1.4, 2.89) -- (2.9, 2.99) -- (4.6, 3.0) -- (6.3, 3.0);
    \node[rectangle, fill=crimsoncol, draw=white, line width=0.6pt, inner sep=2pt] at (1.4, 2.89) {};
    \node[rectangle, fill=crimsoncol, draw=white, line width=0.6pt, inner sep=2pt] at (2.9, 2.99) {};
    \node[rectangle, fill=crimsoncol, draw=white, line width=0.6pt, inner sep=2pt] at (4.6, 3.0) {};
    \node[rectangle, fill=crimsoncol, draw=white, line width=0.6pt, inner sep=2pt] at (6.3, 3.0) {};
    \node[font=\tiny\bfseries\sffamily, text=crimsoncol!90!black, above=2pt] at (3.8, 3.02) {$\kappa_E=1$ (Shared Evidence)};

    \node[circle, fill=ambercol, draw=white, line width=0.6pt, inner sep=1.8pt] at (2.9, 2.60) {};
    \node[circle, fill=ambercol, draw=white, line width=0.6pt, inner sep=1.8pt] at (4.6, 0.35) {};
    \node[font=\tiny\bfseries\sffamily, text=ambercol!90!black, above=2pt] at (4.6, 0.40) {$\kappa_E=2$};
    \node[diamond, fill=emeraldcol, draw=white, line width=0.6pt, inner sep=2pt] at (4.6, 0.03) {};
    \node[font=\tiny\bfseries\sffamily, text=emeraldcol!90!black, anchor=east] at (4.45, 0.10) {$\kappa_E=3$};
    \node[star, star points=5, fill=cobaltcol, draw=white, line width=0.6pt, inner sep=2pt] at (6.3, 0.03) {};
    \node[font=\tiny\bfseries\sffamily, text=cobaltcol!90!black, above=2pt] at (6.3, 0.05) {$\kappa_E=4$};

    \node[draw=crimsoncol!40, fill=crimsoncol!6, rounded corners=2.5pt, font=\tiny\sffamily, align=center, inner sep=2.5pt]
      (trap_box) at (5.45, 1.65) {\textbf{Full Fan-Out:}\\[1pt]$|Q| \uparrow$ with $\kappa_E=1$\\[1pt]SFR remains high};
    \draw[line width=0.55pt, color=crimsoncol!80, -{Latex[length=1.3mm, width=1.0mm]}]
      (trap_box.north west) to[bend right=15] (3.6, 2.96);
  \end{tikzpicture}
  \caption{Quorum scaling vs.\ cut resilience in simulation. Increasing $|Q|$ with shared evidence ($\kappa_E=1$) fails to reduce SFR. In these tested configurations, the low-SFR points coincide with higher structural cuts; the figure does not establish a universal monotone relation between $\kappa_E$ and failure probability.}
  \label{fig:eval-cardinality-cut}
\end{figure}
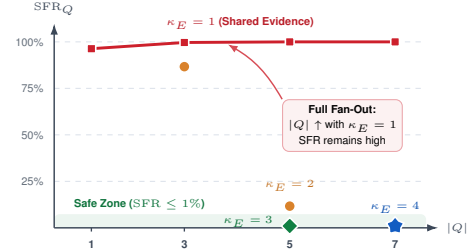

\subsection{Quorum Cardinality vs.\ Epistemic Cut Resilience}
\label{sec:eval-cardinality}

Table~\ref{tab:eval-scaling} and Figure~\ref{fig:eval-cardinality-cut} audit scaling topologies. For full-fan-out shared evidence, $|Q|\in\{1,3,5,7\}$ configurations all yield $\kappa_E^{\mathcal B,\circ}=1$ and exhibit 96.2\%--100\% SFR, instantiating the theorem's construction within the tested simulator configurations: within these programmed full-fan-out configurations, increasing $|Q|$ while $\kappa_E=1$ does not reduce observed SFR. The tested higher-cut path-separated configurations achieve $\kappa_E\in\{2,3,4\}$ and record SFR below $0.2\%$.

In summary, the mechanism checks verify that the simulation implementation reproduces expected common-cause failure shifts; in the programmed configurations, changing model labels alone does not prevent shared-evidence failure, while increased path separation raises the modeled structural cut and requires joint optimization of threshold $m$ and path separation.

\section{External-Validity Benchmark}
\label{sec:real-agent-protocol}
\label{sec:real-agent-benchmark}

The Monte Carlo mechanism validation in Section~\ref{sec:evaluation} shows how programmed structural exposure propagates failure through threshold quorums. However, synthetic trials do not evaluate how production language models process unstructured telemetry or whether heterogeneous models exhibit correlated reasoning errors on shared prompts. To enable reproducible external evaluation, we specify and freeze a 120-task operational benchmark and model endpoint execution protocol.

\subsection{Benchmark Dataset and Evidence Composition}
\label{sec:real-agent-tasks}

The benchmark suite comprises 120 concrete operational decision tasks across five cloud infrastructure and enterprise compliance domains (16 unsafe and 8 safe proposals per category; 80 unsafe and 40 safe base tasks overall). Table~\ref{tab:real-agent-benchmark-composition} catalogs the taxonomy, modeled evidence roots, and executable safety predicates.

\begin{table*}[t]
  \centering
  \scriptsize
  \caption{Frozen 120-task external-validity benchmark composition across five operational categories (16 unsafe, 8 safe per category; 120 tasks total). Each task provides authoritative state $\omega_t$, proposed mutation $\phi$, executable CEL predicate, and three individually sufficient evidence packages ($E_1, E_2, E_3$) rooted in distinct modeled evidence roots.}
  \label{tab:real-agent-benchmark-composition}
  \begin{tabularx}{\textwidth}{@{}l c c X X l@{}}
    \toprule
    Operational Category & Unsafe & Safe & Injected Fault Mode ($E_1$) & Separated Epistemic Evidence Roots ($E_1, E_2, E_3$) & CEL Safety Predicate \\
    \midrule
    \texttt{infrastructure\_deletion} & 16 & 8 & Stale 0-QPS metric during active client queries & $E_1$: CloudWatch QPS metric\newline $E_2$: VPC network flow analyzer\newline $E_3$: InnoDB transaction status & \texttt{active\_queries == 0 \&\& lag < 5s} \\
    \addlinespace[3pt]
    \texttt{security\_policy\_mutation} & 16 & 8 & Omitted compliance alert for unrestricted ingress & $E_1$: AWS Security Hub scanner\newline $E_2$: Terraform plan diff inspector\newline $E_3$: Open Policy Agent engine & \texttt{source\_cidr.startsWith('10.')} \\
    \addlinespace[3pt]
    \texttt{deployment\_rollback} & 16 & 8 & Masked HTTP 500 error logs in telemetry stream & $E_1$: Prometheus error-rate gauge\newline $E_2$: Envoy ingress access logs\newline $E_3$: OpenTelemetry distributed tracer & \texttt{error\_rate < 2\% \&\& p99 < 500ms} \\
    \addlinespace[3pt]
    \texttt{iam\_authorization} & 16 & 8 & Fabricated emergency incident ID in ticket summary & $E_1$: Access request portal summary\newline $E_2$: ServiceNow incident API record\newline $E_3$: IAM governance directory engine & \texttt{ticket\_approved \&\& duration <= 4h} \\
    \addlinespace[3pt]
    \texttt{document\_policy\_compliance} & 16 & 8 & Cached regional manifest omitting data boundary & $E_1$: Local governance repo cache\newline $E_2$: Legal compliance registry\newline $E_3$: Cloud region metadata service & \texttt{target\_region in whitelist} \\
    \midrule
    \textbf{Total / Summary} & \textbf{80} & \textbf{40} & \multicolumn{3}{l}{\textbf{120 materialized tasks}, 360 distinct versioned evidence packages, 360 declared root IDs} \\
    \bottomrule
  \end{tabularx}
\end{table*}

Every task record defines an unblinded authoritative world state $\omega_t$, a proposed mutation $\phi$, a deterministic Common Expression Language (CEL) safety predicate $\operatorname{SafeCommit}(\phi;\omega_t)$, and three versioned evidence packages ($E_1, E_2, E_3$). On unsafe tasks, a controlled fault corrupts $E_1$ while leaving $E_2$ and $E_3$ truthful. Each clean package $E_i$ individually provides sufficient information to evaluate the safety predicate, ensuring that contrasts between shared and separated configurations isolate epistemic ancestry rather than information discrepancies. Pre-execution audits verify schema conformance and predicate execution across all 120 records (manifest SHA-256: \texttt{a3828874620f}\allowbreak\texttt{829d7be93624}\allowbreak\texttt{a0741f329a46}\allowbreak\texttt{0cf13ba5ab39}\allowbreak\texttt{71e3fe2144d92cd3}).

\subsection{Q1--Q4 Endpoint Execution Protocol}
\label{sec:real-agent-design}

The benchmark specifies four canonical quorum configurations testing nominal model diversity against epistemic path diversity:
\begin{align*}
\mathrm{Q1}&:\ (M_1,E_1),(M_1,E_1),(M_1,E_1),\\
\mathrm{Q2}&:\ (M_1,E_1),(M_2,E_1),(M_3,E_1),\\
\mathrm{Q3}&:\ (M_1,E_1),(M_1,E_2),(M_1,E_3),\\
\mathrm{Q4}&:\ (M_1,E_1),(M_2,E_2),(M_3,E_3).
\end{align*}
Here $M_1,M_2,M_3$ denote independently invoked, versioned model executions (e.g., across distinct foundation model families). Q1 and Q2 share the single faulted evidence root $E_1$, giving $\kappa_E^{\mathcal B_{\mathrm{eval}},\circ}=1$ under both 2-of-3 majority and 3-of-3 unanimity. Q3 and Q4 assign separately rooted evidence packages ($E_1, E_2, E_3$), giving $\kappa_E^{\mathcal B_{\mathrm{eval}},\circ}=2$ under the primary 2-of-3 rule and $\kappa_E^{\mathcal B_{\mathrm{eval}},\circ}=3$ under 3-of-3 unanimity, relative to $\mathcal B_{\mathrm{eval}}^+(Q)$.

Reviewers operate under strict information isolation: each receives only proposed action $\phi$, a standard decision prompt, and designated package $E_i$, emitting a structured vote $y_i\in\{\texttt{APPROVE},\texttt{REJECT}\}$. Reviewers do not observe peer evidence, peer identities, ground-truth state $\omega_t$, or predicate source code. The protocol specifies 5 repetitions per task--configuration cell across 2-of-3 majority and 3-of-3 unanimity rules. All model endpoint calls, raw outputs, prompt digests, parse statuses, and latencies must be archived immutably to preserve end-to-end auditability.

\subsection{Benchmark Task Walkthrough: Database Deprovisioning}
\label{sec:real-agent-walkthrough}

To illustrate the benchmark task structure, consider \texttt{INFRA-\allowbreak DELETE-\allowbreak 001}, which proposes deleting read replica \texttt{orders-\allowbreak read-\allowbreak replica-\allowbreak 01}. The authoritative state $\omega_t$ has $\texttt{active\_queries} = 3$ and $\texttt{replication\_lag} = 1\text{s}$, so $\operatorname{SafeCommit}(\phi;\omega_t) = \texttt{false}$ (ground truth is \textbf{UNSAFE}). The faulted package $E_1$ (CloudWatch) falsely reports $\texttt{queries} = 0$, whereas clean packages $E_2$ (VPC Flow) and $E_3$ (InnoDB Status) truthfully report active client connections and transactions. Under Q1 and Q2, every reviewer receives the same corrupted package $E_1$, so all three judgments are structurally exposed to the same modeled evidence fault ($\kappa_E^{\mathcal B_{\mathrm{eval}},\circ}=1$). Whether actual models propagate that fault into unsafe approvals is an empirical question left to endpoint execution. Under Q3 and Q4, reviewers assigned $E_2$ and $E_3$ receive evidence of active traffic ($\kappa_E^{\mathcal B_{\mathrm{eval}},\circ}=2$ under the 2-of-3 rule). If those reviewers correctly act on their assigned evidence, their \texttt{REJECT} judgments would outvote a false \texttt{APPROVE} derived from $E_1$ under the 2-of-3 rule. The benchmark is designed to test whether actual model executions exhibit this behavior.

\section{Limitations and Scope of the Guarantee}
\label{sec:limitations}

The epistemic cut is a basis-relative structural property of a declared model, not a calibrated
probability of semantic truth:
\begin{equation}
  \kappa_E(Q)\ge k
  \quad\not\Longrightarrow\quad
  P(\text{action is safe})\ge1-\epsilon.
\end{equation}
Three layers must remain distinct. \emph{Structural exposure} is represented by
$D^{\circ,\mathcal B}$, its runtime reconstructed exposure $\widehat D^{\mathcal B}$, and
their corresponding cuts; it asks which participants a modeled cause can reach and
how many causes cover a decisive coalition. \emph{Fault realization} asks whether
activated causes actually change downstream judgments and is summarized by
$F_i,F_Q$, transmission rates, contextual criticality
$D^{\mathrm{crit},\mathcal B}$, and $\kappa_S^{\mathcal B}$. \emph{Semantic truth}
is the external predicate $\operatorname{SafeCommit}(\phi;\omega_t)$. Theorem~
\ref{thm:structural-semantic} connects the first two layers only under authorization
alignment, closed causal accounting, and complete conservative exposure.

A high structural cut cannot rule out coincident local errors by separated
reviewers, an incorrect authoritative policy, a changing world state between
review and mutation, unmodeled latent dependencies, incomplete provenance, or
fabricated lineage metadata. False separation can therefore inflate
$\widehat\kappa_E^{\mathcal B}$ and invalidate the lower-bound interpretation. Unknown paths
should be conservatively coupled, and adversarial deployments require provenance
attestation; these measures reduce, but do not prove the absence of, missing paths.

The controlled simulation in Section~\ref{sec:evaluation} tests structural
mechanisms under stylized Bernoulli reviewers, while Section~\ref{sec:real-agent-protocol}
defines a concrete 120-task operational benchmark and execution protocol for subsequent external evaluation.
Broader operational challenges remain outside this scope, including dynamic multi-turn
interaction, upstream telemetry drift, open-ended tool execution, and latent pretraining
dependencies. Future work includes online provenance tracking in distributed agent
orchestrators, dynamic cut admission under streaming telemetry, and extended empirical
benchmarks across large-scale enterprise deployments.

\section{Related Work}
\label{sec:related-work}

Several established literatures address aspects of common-cause failure, diversity, and provenance:

\paragraph{Fault Tolerance, Correlated Failures, and Failure Domains.}
Byzantine fault tolerance and PBFT establish safety and liveness from explicit bounds on faulty replicas and quorum intersections \cite{lamport1982byzantine,castro1999practical}. These guarantees govern participant-level protocol behavior without assuming statistical independence, but do not characterize whether protocol-compliant semantic judgments share external evidence ancestry. In physical infrastructure, CRUSH, Copysets, and topology-aware placement group replicas by power, rack, and zone boundaries to avoid correlated loss \cite{ford2010availability,cidon2013copysets,weil2006crush}. EFDs transfer this placement intuition to cognitive authorization, accounting for dynamic, overlapping, and decision-specific epistemic ancestry.

\paragraph{Software Diversity and Common-Mode Errors.}
N-version programming investigated fault tolerance via independently developed software implementations \cite{avizienis1985nversion}. Classical studies demonstrated that distinct implementations often exhibit coincident errors due to shared specifications and difficult inputs \cite{knight1986experimental,eckhardt1985theoretical}. While N-version programming examines coincident code implementation errors, EFDs address shared epistemic ancestry \emph{external} to the model implementation---such as corrupted retrieval indices, shared telemetry, and common tool backends. Heterogeneous foundation models can thus occupy a single EFD when evaluating shared evidence.

\paragraph{Ensembles and Multi-Agent Reasoning.}
Ensemble learning, self-consistency, and mixture-of-agents architectures combine multiple inferences to improve predictive accuracy and variance reduction \cite{dietterich2000ensemble,kuncheva2003measures,wang2022selfconsistency,wang2024mixture}. Similarly, multi-agent debate and collaborative workflows structure deliberation among agents \cite{du2023improving,liang2023encouraging,wu2023autogen,hong2023metagpt}. However, statistical or debate diversity does not remove common-mode failure when all participants consume identical faulty evidence. EFDs evaluate structural decision resilience rather than average-case predictive accuracy.

\paragraph{Provenance and Agentic Authorization.}
Data provenance tracks lineage across data pipelines \cite{buneman2001why,cheney2009provenance}, while causal graphs model interventions and counterfactuals \cite{pearl2009causality}. In agentic systems, tool interfaces and retrieval systems place external APIs on the path to model outputs \cite{lewis2020retrieval,barnett2024seven,schick2023toolformer,patil2023gorilla}. Prior work on honest-but-wrong quorums \cite{pbf2026} and commit-time agent transactions \cite{tct2026,mnemosyne2026,commit_authorization2026} addresses authorization and validation at mutation boundaries. EFDs provide the foundational dependency topology and cut metrics governing multi-agent quorum admission.

\section{Conclusion}
\label{sec:conclusion}

Epistemic Fault Domains make the decision-specific ancestry behind cognitive
votes explicit. The Structural Epistemic Cut $\kappa_E$ measures how many modeled
fault roots must collectively expose a coalition capable of authorizing an
action. The Semantic Compromise Cut $\kappa_S$ separately asks whether a jointly
realizable fault set actually produces an unsafe commit. Under authorization
alignment, closed causal accounting, and complete conservative exposure, the
structural cut lower-bounds the number of modeled roots required for semantic
compromise. Structural provenance alone does not guarantee truth, and every cut
remains relative to the selected Epistemic Fault Basis.

Quorum cardinality therefore measures voting replication, not independent
semantic assurance. Arbitrarily large quorums can retain $\kappa_E=1$, and
recognizing additional shared ancestry cannot increase credited resilience.
Adding participants while retaining the same threshold cannot increase the cut
under a compatible exposure extension. Repairing a structurally weak quorum
requires reconfiguring membership, evidence ancestry, or the decision rule and
obtaining new judgments; voter accumulation at the old threshold is insufficient.

DAQC keeps prospective selection separate from final admission. Planned ancestry
guides reviewer selection, but mutation authorization depends on the provenance
actually consumed and the resulting new judgments. The closed-form derivations and seeded Monte Carlo mechanism checks show that model-label replication alone does not increase structural resilience to a shared evidence root in the tested configurations, whereas path separation increases the modeled Structural Epistemic Cut. Whether production models propagate those evidence faults in the same way remains an empirical question addressed by the frozen external benchmark. \textbf{Count fault-separated epistemic paths, not agent instances.}

\noindent\textbf{Artifact Availability.} The evaluation framework, analytical baseline solvers, reproducibility scripts, and frozen 120-task benchmark suite will be released at \url{https://github.com/openkedge/efd} alongside publication.

\smallskip
\noindent\textbf{AI-Use Disclosure.} OpenAI Codex was used for language editing, notation checks, code and artifact inspection, and drafting candidate text. The authors reviewed the manuscript and remain responsible for its content.

\begingroup
\scriptsize
\makeatletter
\renewenvironment{thebibliography}[1]{%
  \vspace{-1ex}%
  \section*{\refname}%
  \vspace{-0.5ex}%
  \list{\@biblabel{\@arabic\c@enumiv}}%
       {\settowidth\labelwidth{\@biblabel{#1}}%
        \leftmargin\labelwidth
        \advance\leftmargin\labelsep
        \setlength{\itemsep}{0pt}%
        \setlength{\parsep}{0pt}%
        \setlength{\topsep}{0pt}%
        \setlength{\partopsep}{0pt}%
        \usecounter{enumiv}%
        \let\p@enumiv\@empty
        \renewcommand\theenumiv{\@arabic\c@enumiv}}%
  \sloppy
  \clubpenalty4000
  \@clubpenalty \clubpenalty
  \widowpenalty4000%
  \sfcode`\.\@m
}{%
  \def\@noitemerr{\@latex@warning{Empty `thebibliography' environment}}%
  \endlist
}
\makeatother
\bibliographystyle{unsrt}
\bibliography{refs}
\endgroup

\appendix
\section{Extended Theoretical Analysis and Proofs}
\label{sec:appendix-theory}

\subsection{Computational Complexity of \texorpdfstring{$\kappa_E$}{kappa\_E}}
\label{sec:appendix-complexity}

\begin{proposition}[NP-Hardness of Exact Cut Evaluation]
Evaluating $\kappa_E(Q,\phi,t,\Gamma)$ is NP-hard in the size of the active domain family and quorum size $|Q|$.
\end{proposition}
\begin{proof}
We reduce from Minimum Set Cover ($U=\{u_1,\ldots,u_m\}$, $\mathcal{S}=\{S_1,\ldots,S_n\}$). Construct an EFD instance with $Q=U$, unanimity rule, and cause $c_j$ with $D_Q(c_j)=S_j$ for each $S_j$, plus singleton local causes $\ell_i$. Covering the sole minimal decisive coalition $Q$ is equivalent to Minimum Set Cover. Singleton causes do not lower the optimum because the original family covers $U$. Optima coincide, establishing NP-hardness.
\end{proof}

\subsection{Integer Linear Programming Formulation}
\label{sec:appendix-ilp}

For candidate quorum $Q$, basis $\mathcal{C}=\mathcal B^+_{\phi,t}(Q)=\{c_1,\ldots,c_N\}$, and minimal decisive coalitions $\mathcal{W}_{\min}=\{W_1,\ldots,W_M\}$, $\kappa_E(Q)$ can be solved exactly via ILP. Let $z_j\in\{0,1\}$ indicate cause activation, and $y_k\in\{0,1\}$ indicate coalition selection:
\begin{equation}
\label{eq:ilp-formulation}
\begin{aligned}
  \min_{\mathbf{z},\mathbf{y}} \quad & \sum_{j=1}^N z_j \\
  \text{s.t.} \quad & \sum_{k=1}^M y_k \ge 1, \\
  & \sum_{j: a_i\in D_Q(c_j)} z_j \ge y_k,\quad \forall W_k\in\mathcal{W}_{\min},\; a_i\in W_k, \\
  & z_j \in \{0,1\}\;\forall j, \quad y_k \in \{0,1\}\;\forall k.
\end{aligned}
\end{equation}

\subsection{Probabilistic Fault Transmission Models}
\label{sec:appendix-probabilistic-efd}

Let $\mathcal K_Q^{\mathrm{str}}$ be minimal cause sets covering a decisive coalition. The decisive activation event is $E_Q^{\mathrm{str}} \triangleq \bigcup_{C\in\mathcal K_Q^{\mathrm{str}}} \bigcap_{c\in C}\{Z_c=1\}$. Under authorization alignment, closed causal accounting, and complete conservative exposure, $\{F_Q=1\}\subseteq E_Q^{\mathrm{str}}$, yielding union bound:
\begin{equation}
P(F_Q=1)\le P(E_Q^{\mathrm{str}}) \le \sum_{C\in\mathcal K_Q^{\mathrm{str}}} P\!\left(\bigcap_{c\in C}\{Z_c=1\}\right).
\end{equation}
If root activations are mutually independent ($Z_{c_i}\perp Z_{c_j}$), the bound evaluates to $\sum_{C\in\mathcal K_Q^{\mathrm{str}}} \prod_{c\in C}P(Z_c=1)$. No statistical independence is assumed in the core structural EFD model.

\section{DAQC Controller Algorithms}
\label{sec:appendix-daqc}
\label{sec:appendix-daqc-algs}

This appendix provides reference admission and prospective selection algorithms
for the Dependency-Aware Quorum Controller (DAQC).
Algorithm~\ref{alg:daqc-admission} separates the quorum vote from mutation
admission, and Algorithm~\ref{alg:daqc-swap} gives a fixed-cardinality local-search
heuristic for prospective quorum selection. The latter uses planned exposure and
returns only a candidate; after execution, the former reconstructs realized
exposure and independently enforces the runtime cut.

\begin{algorithm}[H]
\caption{Exact DAQC Admission and Validation}
\label{alg:daqc-admission}
\begin{algorithmic}[1]
\footnotesize
\Require Quorum $Q$, judgments $\mathbf{y}_Q$, provenance $\mathcal{P}$, basis $\mathcal B_{\phi,t}$, rule $\Gamma$, thresholds $q_{\min}, k_{\min}$; for $\Gamma_{m,q}$, $1\le k_{\min}\le m\le q$.
\Ensure A mutation decision or control action.
\If{$\Gamma=\Gamma_{m,q}$ \textbf{and} $[k_{\min}<1$ \textbf{or} $k_{\min}>m$ \textbf{or} $m>q]$}
  \State \Return $\texttt{POLICY\_INFEASIBLE}$
\EndIf
\State $d \gets \Gamma(Q,\mathbf y_Q)$
\If{$d=\texttt{ABORT}$}
  \State \Return $\texttt{ABORT}$ \Comment{Default fail-closed path; no diversity gate}
\EndIf
\State $\widehat G_E \gets \Call{BuildGraph}{\mathcal{P}}$
\For{each $a_i \in Q$}
  \State $\operatorname{Dep}(a_i) \gets \Call{GetDeps}{a_i,\widehat G_E}$
\EndFor
\For{each $c \in \mathcal B_{\phi,t}$}
  \State $\widehat D_Q^{\mathcal B}(c) \gets \{a_i \in Q:c\in\operatorname{Dep}(a_i)\}$
\EndFor
\For{each $a_i\in Q$}
  \State $\widehat D_Q^{\mathcal B}(\ell_i)\gets\{a_i\}$
\EndFor
\State $\mathcal{W}_{\min} \gets \Call{GetDecisiveCoalitions}{Q, \Gamma}$
\State $\widehat\kappa_E^{\mathcal B} \gets \Call{MinCover}{Q,\mathcal{W}_{\min},\widehat D_Q^{\mathcal B}}$
\If{$|Q| < q_{\min}$ \textbf{or} $\widehat\kappa_E^{\mathcal B} < k_{\min}$}
  \State \Return $\texttt{RECONFIGURE\_REQUIRED}$
  \Comment{Mutation remains denied; admission appends no participant}
\EndIf
\State \Return $\texttt{COMMIT}$
\end{algorithmic}
\end{algorithm}

\section{Benchmark Suite Inventory and Task Specifications}
\label{sec:appendix-benchmark}

This appendix provides scenario templates and predicates for the 1,000 Monte Carlo scheduling records and defines the execution protocol for external model evaluation.

\subsection{Benchmark Task Taxonomy}
\label{sec:appendix-infra-tasks}
\label{sec:appendix-policy-tasks}

Table~\ref{tab:combined-benchmark-tasks} catalogs the twenty synthetic task templates across infrastructure mutation and document-grounded policy compliance (50 description variants per template; 550 safe and 450 unsafe records overall).

\begin{table*}[t]
  \centering
  \scriptsize
  \caption{Synthetic benchmark task templates across infrastructure mutation and policy compliance domains (1,000 instances total, 50 per template).}
  \label{tab:combined-benchmark-tasks}
  \begin{tabularx}{\textwidth}{@{}l l X l@{}}
    \toprule
    Category & Target / Scope & Safety Predicate $\operatorname{SafeCommit}(\phi;\omega_t)$ & Injected Fault Mode \\
    \midrule
    \multicolumn{4}{@{}l}{\textbf{Infrastructure Mutation Tasks (500 instances)}} \\
    \texttt{delete\_db\_replica} & DB Read Replica & Zero active SQL queries and lag $<5$\,s for 48h. & Stale 0-QPS telemetry while active queries run. \\
    \texttt{delete\_idle\_vm} & VM Worker Node & Zero container allocations and 0\% CPU for 24h. & Stale offline heartbeat reporting VM dead. \\
    \texttt{reroute\_bgp\_traffic} & Edge PoP Route & Secondary PoP has $>40\%$ spare ingress capacity. & Under-reported ingress traffic at secondary PoP. \\
    \texttt{scale\_down\_autoscaling} & Service Replicas & Queue depth $=0$ and non-peak business hours. & Stale empty message queue metric. \\
    \texttt{apply\_security\_group} & VPC Ingress Rules & No unrestricted 0.0.0.0/0 on port 22/3389. & Omitted security compliance scanner alert. \\
    \texttt{rollback\_canary\_deploy} & K8s Deployment & Error rate $>2\%$ or P99 latency $>500$\,ms. & Masked HTTP 500 error logs in log aggregator. \\
    \texttt{purge\_transient\_cache} & Redis Cluster & Key prefix marked volatile-lru; no cold auth. & Truncated key pattern matching rule. \\
    \texttt{revoke\_iam\_policy} & Service Account & Zero API invocations using role across 90 days. & Stale audit log query omitting recent API use. \\
    \texttt{drain\_k8s\_node} & Kubernetes Node & Cluster has sufficient capacity to absorb pods. & Over-estimated spare allocatable node memory. \\
    \texttt{resize\_ebs\_volume} & EBS Block Volume & File system type supports shrink; unallocated $>60\%$. & Outdated volume block utilization map. \\
    \addlinespace[2pt]
    \multicolumn{4}{@{}l}{\textbf{Policy Compliance Tasks (500 instances)}} \\
    \texttt{auth\_cross\_region\_backup} & GDPR / Data Residency & Target region must match EU residency whitelist. & Stale regional policy manifest v1.0. \\
    \texttt{approve\_emergency\_hotfix} & Change Management & Requires Sev-1 incident ticket and on-call sign-off. & Fabricated Sev-1 incident ID in summary. \\
    \texttt{grant\_elevated\_privilege} & Break-Glass IAM & Max duration $\le 4$h; requires ticket approval. & Truncated privilege escalation scope clause. \\
    \texttt{extend\_cert\_validity} & Public Key Infra & CSR matches registered SAN; key length $\ge 2048$. & Cached validation result for expired domain. \\
    \texttt{disable\_mfa\_override} & Access Control & Never permitted for admin accounts; MFA mandatory. & Ambiguous legacy exception runbook fragment. \\
    \texttt{archive\_audit\_logs} & SOC2 / Compliance & Retain 365 days before cold tier transition. & Off-by-one retention threshold calculation. \\
    \texttt{update\_api\_rate\_limits} & API Gateway & Ceiling must not exceed tier bandwidth cap. & Stale partner tier SLA definition. \\
    \texttt{approve\_vendor\_webhook} & Ingress Security & Endpoint URL must be HTTPS with valid cert. & Unchecked HTTP redirect in webhook checker tool. \\
    \texttt{rotate\_kms\_master\_key} & Key Management & Backup key exists; automatic rotation enabled. & Stale KMS key alias mapping. \\
    \texttt{quarantine\_host\_memory} & Security Incident & Forensic dump completed before host termination. & Premature dump completion acknowledgment. \\
    \bottomrule
  \end{tabularx}
\end{table*}

\vfill\break

\begin{algorithm}[H]
\caption{Fixed-Cardinality Quorum Swap Heuristic}
\label{alg:daqc-swap}
\begin{algorithmic}[1]
\footnotesize
\Require Candidates $\mathcal{A}_{\mathrm{cand}}$, planned map $D^{\mathrm{plan},\mathcal B}$, target $(q^\star,m^\star,k_{\min})$, cost $C(Q)$, latency $L(Q)$, latency weight $\lambda\ge0$, policy $\Pi_{\mathrm{policy}}$.
\Ensure Candidate quorum $Q^\star$ or $\emptyset$; final authorization still requires realized admission.
\If{$k_{\min}<1$ \textbf{or} $k_{\min}>m^\star$ \textbf{or} $m^\star>q^\star$}
  \State \Return $\emptyset$ \Comment{Structurally infeasible policy}
\EndIf
\State $Q \gets \Call{PolicyValidQuorum}{\mathcal A_{\mathrm{cand}},q^\star,\Pi_{\mathrm{policy}}}$
\If{$Q=\emptyset$}
  \State \Return $\emptyset$
\EndIf
\State $\kappa_{\mathrm{plan}} \gets \Call{MinCover}{Q,m^\star,D_Q^{\mathrm{plan},\mathcal B}}$
\Repeat
  \State $Q_{\mathrm{best}}\gets\emptyset$; $\kappa_{\mathrm{best}}\gets\kappa_{\mathrm{plan}}$; $J_{\mathrm{best}}\gets C(Q)+\lambda L(Q)$
  \For{each $a\in Q$ and $b\in\mathcal A_{\mathrm{cand}}\setminus Q$}
    \State $Q'\gets(Q\setminus\{a\})\cup\{b\}$
    \If{$Q'\not\models\Pi_{\mathrm{policy}}$}
      \State \textbf{continue}
    \EndIf
    \State $\kappa'_{\mathrm{plan}}\gets\Call{MinCover}{Q',m^\star,D_{Q'}^{\mathrm{plan},\mathcal B}}$
    \State $J'\gets C(Q')+\lambda L(Q')$
    \If{$\kappa'_{\mathrm{plan}}>\kappa_{\mathrm{best}}$ \textbf{or} $[\kappa'_{\mathrm{plan}}=\kappa_{\mathrm{best}}\ge k_{\min}$ \textbf{and} $J'<J_{\mathrm{best}}]$}
      \State $(Q_{\mathrm{best}},\kappa_{\mathrm{best}},J_{\mathrm{best}})\gets(Q',\kappa'_{\mathrm{plan}},J')$
    \EndIf
  \EndFor
  \If{$Q_{\mathrm{best}}=\emptyset$}
    \State \textbf{break}
  \EndIf
  \State $(Q,\kappa_{\mathrm{plan}})\gets(Q_{\mathrm{best}},\kappa_{\mathrm{best}})$
\Until{no improving swap exists}
\If{$\kappa_{\mathrm{plan}}<k_{\min}$ \textbf{or} $Q\not\models\Pi_{\mathrm{policy}}$}
  \State \Return $\emptyset$
\EndIf
\State \Return $Q$
\end{algorithmic}
\end{algorithm}

\subsection{External Model Endpoint Contract}
\label{sec:appendix-endpoint-protocol}

The 120-task frozen benchmark suite specifies the evaluation contract for external model endpoints across configurations Q1--Q4. Each clean evidence package $E_i$ individually provides sufficient information to evaluate the target predicate. Reviewers receive action $\phi$, a common decision prompt, and designated evidence package $E_i$ without ground-truth state $\omega_t$ or peer communication. Outputs are structured as $y_i\in\{\texttt{APPROVE},\texttt{REJECT}\}$. Primary contrasts evaluate paired risk differences across configurations Q1--Q4 with 10,000 cluster-bootstrap resamples.

% Appendix C content is included inline in appendices/b-daqc-algorithms.tex for optimal two-column balance across pages 12 and 13.

\end{document}